\pdfoutput=1 

\documentclass{article}

\usepackage{arxiv-style}

\usepackage[utf8]{inputenc} 
\usepackage[T1]{fontenc}    
\usepackage{csquotes}
\usepackage{placeins, float}
\usepackage{hanging}
\usepackage{graphicx}
\usepackage{setspace}
\usepackage{multirow, booktabs, subcaption, colortbl, longtable, caption}
\usepackage{lmodern}
\usepackage[bottom]{footmisc}
\usepackage{xcolor}
\usepackage{pdflscape, afterpage, rotating}
\usepackage{tabularx, threeparttable, makecell}
\newcolumntype{P}[1]{>{\RaggedRight\hspace{0pt}}p{#1}}
\newcolumntype{C}[1]{>{\Centering\hspace{0pt}}m{#1}}
\usepackage[edges]{forest}
\usepackage{tikz}
\usetikzlibrary{backgrounds,tikzmark, calc,arrows,shapes,positioning,shadows,trees,mindmap,arrows.meta}
\colorlet{linecol}{black!75}
\usepackage{amsmath, amsthm, amssymb, mathtools, nccmath}
\newtheorem{theorem}{Theorem}
\newtheorem{corollary}[theorem]{Corollary}
\usepackage{wrapfig}
\usepackage{comment}
\usepackage{xspace}
\usepackage{array}
\usepackage{ragged2e}

\usepackage{tcolorbox}

\definecolor{red}{HTML}{F03D2D}
\definecolor{skyblue}{HTML}{90DDF0}
\definecolor{green}{HTML}{C8D96F}
\definecolor{orange}{HTML}{EF8A17}
\definecolor{yellow}{HTML}{F5C900}
\definecolor{purple}{HTML}{BA42C0}
\definecolor{teal}{HTML}{17BEBB}
\definecolor{greyblue}{HTML}{9BAFD9}
\definecolor{bluegreen}{HTML}{9FD8CB}
\definecolor{pink}{HTML}{D14081}
\usepackage{hyperref}
\usepackage{orcidlink}

\hypersetup{
    colorlinks=true,
    linkcolor=blue,
    citecolor=blue,
    urlcolor=blue,
    pdftitle={Adaptive Regularization via Extreme Value Distributions for Gaussian Graphical Models},
    pdfauthor={Alexander P. Christensen, Jeongwon Choi, Haoyi Yang, Lingzhou Xue},
    pdfsubject={Gaussian graphical models, regularization, extreme value theory},
    pdfkeywords={Gaussian graphical models, regularization, extreme value theory}
}
\usepackage{bookmark}
\IfFileExists{xurl.sty}{\usepackage{xurl}}{} 
\definecolor{eganetbg}{HTML}{D5896F}
\definecolor{eganettxt}{HTML}{4A1F08}
\definecolor{qgraphbg}{HTML}{DAB785}
\definecolor{qgraphtxt}{HTML}{4A320A}
\definecolor{ggmmodbg}{HTML}{70A288}
\definecolor{ggmmodtxt}{HTML}{1C3D2C}
\definecolor{nonregbg}{HTML}{72A6B1}
\definecolor{nonregtxt}{HTML}{1A3B42}
\definecolor{ncvbg}{HTML}{83628A}
\definecolor{ncvtxt}{HTML}{2E1535}
\colorlet{eganetbg}{eganetbg!50!white}
\colorlet{qgraphbg}{qgraphbg!50!white}
\colorlet{ggmmodbg}{ggmmodbg!50!white}
\colorlet{nonregbg}{nonregbg!50!white}
\colorlet{ncvbg}{ncvbg!50!white}
\definecolor{anybg}{HTML}{EFEFEF}

\usepackage[natbibapa]{apacite}

\ifLuaTeX
\usepackage[bidi=basic,shorthands=off]{babel}
\else
\usepackage[bidi=default,shorthands=off]{babel}
\fi
\title{Adaptive Regularization via Extreme Value Distributions for Gaussian Graphical Models}

\author{
   Alexander P.~Christensen$^{*}$ \orcidlink{0000-0002-9798-7037} \\
  Department of Psychology and Human Development \\
  Vanderbilt University \\
  \And
  Jeongwon~Choi$^{*}$ \orcidlink{0000-0001-6087-2124} \\
  Department of Psychology and Human Development \\
  Vanderbilt University \\
  \And
  Haoyi~Yang$^{*}$ \\
  Department of Statistics \\
  The Pennsylvania State University \\
  \And
  Lingzhou~Xue \orcidlink{0000-0002-8252-0637} \\
  Department of Statistics \\
  The Pennsylvania State University \\
}

\begin{document}

\maketitle

\begin{abstract}
Edge selection in Gaussian graphical models is fundamentally a variable selection problem where pairwise relationships determine construct validity and variable importance in psychological networks. In psychology, network estimation relies predominantly on \(\ell_1\) regularization where uniform shrinkage systematically underestimates edge and centrality parameters. Alternative penalties overcome this bias but rely on fixed hyperparameters that do not adapt to the signal in the data. We develop a family of data-adaptive regularization penalties grounded in extreme value theory. Across 290 empirical psychological datasets, we show that absolute partial correlations are well-described by the Weibull distribution. Using this empirical regularity, we derive Weibull, Gumbel, and Exponential penalties that approximate \(\ell_0\) penalization and calibrate their hyperparameters to each dataset's noise floor. We formally prove the asymptotic properties of their static forms and conduct a large-scale simulation spanning two network topologies and various sample sizes (\(N\) = 100--10,000), demonstrating that their adaptive forms maintain high specificity while accumulating sensitivity as sample size increases with low parameter bias and high rank-order centrality congruence relative to field standards. Empirically, method choice alone determined centrality rankings at sample sizes typical in psychology. Of the three adaptive penalties, Weibull is recommended given the interpretability of its parameters.
\end{abstract}


\let\thefootnote\relax\footnotetext{$^{*}$These authors contributed equally to this work.

Correspondence concerning this article should be addressed to Alexander P. Christensen, 230 Appleton Place \#552, Nashville, TN, 37203; \href{mailto:alexander.christensen@vanderbilt.edu}{\color{blue}alexander.christensen@vanderbilt.edu} or Lingzhou Xue, 318 Thomas Building, University Park, PA, 16802; \href{mailto:lzxue@psu.edu}{\color{blue}lzxue@psu.edu}.

Acknowledgments: The authors did not preregister the study. We would like to express enormous gratitude to {Huth, Haslbeck, Keetelaar, van Holst, \& Marsman} (\citeyear{huth2025statistical}) for their compilation of 293 empirical network datasets and the researchers who shared their datasets. All data, code, and materials can be found on the \href{https://osf.io/6hxkt}{\color{blue}Open Science Framework} \citep{christensen2026adaptive_osf}. The first draft of the manuscript was written completely without AI assistance. Subsequent revisions used Claude Sonnet 4.6 \citep{anthropic} to provide feedback on proper grammar, ensure correct formula notation, aid in the creation of LaTeX tables, and evaluate portions of text that could be improved for clarity.

The authors made the following contributions. Alexander P. Christensen: Conceptualization, Data Curation, Formal Analysis, Methodology, Software, Validation, Visualization, Writing - Original Draft Preparation, Writing - Review \& Editing; Jeongwon Choi: Conceptualization, Methodology, Software, Writing - Review \& Editing; Haoyi Yang: Validation, Writing - Original Draft Preparation, Writing - Review \& Editing; Lingzhou Xue: Supervision, Validation, Writing - Review \& Editing.}

\newpage

The Gaussian graphical model (GGM; \citealp{lauritzen1996graphical}) is commonly used to define conditional (in)dependence between variables in networks by setting elements of the precision matrix exactly to zero \citep{epskamp2018tutorial}. These conditional relations, referred to as edges, are commonly used to assess construct validity and identify important variables \citep{chambon2026network,van2024measurement}. Construct validity and variable importance (centrality) therefore hinge on whether the estimated network structure (presence and absence of edges) accurately captures the ``true'' network structure. Consequently, edge selection is fundamentally a variable selection problem where the pairwise relationships determine which variables are important.

In psychological networks, network estimation predominantly relies on \(\ell_1\) regularization or least absolute shrinkage and selection operator (LASSO; \citealp{tibshirani1996regression}). This penalty underlies widely used procedures in psychology such as the graphical LASSO (GLASSO; \citealp{friedman2008sparse}) with extended Bayesian information criterion (EBIC) model selection (EBICglasso; \citealp{foygel2010ebic,epskamp2018tutorial}) with extensions to binary \citep{hofling2009estimation,xue2012nonconcave,van2014new,tao2024additive} and ordinal graphical models \citep{guo2015graphical,lee2022estimating}. Despite its popularity, \(\ell_1\) regularization has some limitations. The \(\ell_1\) penalty shrinks all parameters uniformly toward zero. Several studies have demonstrated that this persistent bias leads to greater error in edge weights and produces poor specificity \citep{williams2019nonregularized,wysocki2019penalty}. Non-convex penalties that reduce \(\ell_1\) shrinkage for larger parameters and \(\ell_0\) approximating penalties that target best subset selection often recover networks better than \(\ell_1\) \citep{williams2020beyond}. These alternative penalties, however, rely on fixed hyperparameters that implicitly assume a universal noise floor regardless of the signal in the underlying partial correlation distribution.

The present paper addresses this limitation by developing data-adaptive \(\ell_0\) approximating penalties grounded in extreme value theory (EVT) for GGMs. Rather than fixing hyperparameters a priori, these penalties calibrate to the empirical partial correlation distribution of each dataset, adapting the noise floor to the signal present in the data. We first demonstrate that absolute empirical partial correlations tend to follow a Weibull distribution \citep{christensen2026network}, one of the canonical EVT distributions \citep{fisher1928limiting,gnedenko1943distribution}, then use this regularity to derive the penalty family. We evaluate these penalties against contemporary field standards, EBICglasso and Bayesian GGMs (BGGM; \citealp{huth2024simplifying,williams2021bayesian}), and several non-convex and \(\ell_0\) approximating penalties \citep{williams2020beyond} across simulated networks varying in topology and sample size.

\section{Introduction}

\subsection{Distribution of Empirical Partial Correlations}\label{distribution-of-empirical-partial-correlations}

In a reanalysis of 293 psychological datasets, Huth et al. (\citeyear{huth2025statistical}) found that most edges in psychometric networks had only weak or inconclusive evidential support, with fewer than one-fifth showing strong evidence. This finding raises concerns about the trustworthiness of psychometric network results, but their diverse database also provides an unprecedented opportunity to identify regularities in empirical partial correlation distributions across psychological constructs.

Christensen and Choi (\citeyear{christensen2026network}) reanalyzed their database in search of such regularities, finding that the Weibull distribution best fit the EBICglasso estimated edge weights of these datasets. This analysis, however, examined post-regularization edge weights, already subjected to model selection and shrinkage, rather than the raw partial correlations. Identifying a distribution that underlies the raw partial correlations is theoretically more appropriate because it characterizes the signal prior to any modeling decisions. We therefore extended this analysis to the absolute partial correlations directly and found that the Weibull distribution similarly described the raw partial correlations (see Appendix \ref{appendix-distribution}).

One convenience of the Weibull distribution is that its parameters can be directly connected to empirical statistics of the partial correlations. First, we define the Weibull probability density (PDF) and cumulative distribution function (CDF) functions, respectively:

\[
f(x; k; \gamma) = \frac{k}{\gamma} \bigg( \frac{x}{\gamma} \bigg)^{k - 1} e^{-\big(\frac{x}{\gamma}\big)^k}
\]

\noindent and

\[
F(x; k; \gamma) = 1 - e^{-\big(\frac{x}{\gamma}\big)^k}
\]

\noindent where \(x\) is a vector of values between \([0, \infty]\), \(k\) is the shape parameter and \(\gamma\) is the scale parameter. These definitions represent the two-parameter Weibull distribution, setting the less frequently used third parameter of location to zero (\(\theta = 0\); \citealp{rinne2008weibull}).

The parameters \(\gamma\) and \(k\) can be connected to empirical statistics, which can be used as method of moments estimators \citep{cohen1965maximum}. The mean of the absolute partial correlations, \(\bar{x}\), approximates the Weibull distribution mean: \(\bar{x} \approx \gamma \Gamma \big( 1 + \frac{1}{k} \big)\). When \(k = 1\), the Weibull mean reduces to \(\gamma\) exactly, so the empirical mean directly estimates the scale parameter. The deviation from \(k = 1\) determines the degree to which \(\bar{x}\) diverges from \(\gamma\), with larger deviations in either direction producing greater discrepancy.

The shape parameter, \(k\), can be estimated independently of \(\gamma\) through the coefficient of variation (CV). As Cohen (\citeyear{cohen1965maximum}) showed, the Weibull CV is a function of \(k\) alone,

\[
\frac{\sigma_x}{\bar{x}} \approx \frac{ \bigg[ \Gamma \big(1 + \frac{2}{k}\big) - \bigg( \Gamma \big( 1 + \frac{1}{k} \big) \bigg)^2 \bigg]}{ \Gamma \big(1 + \frac{1}{k} \big)},
\]

\noindent because \(\gamma\) cancels out in the division of the standard deviation (numerator) over the mean (denominator). The inverse of the empirical CV, the signal-to-noise ratio (\(SNR = \frac{1}{CV} = \frac{\bar{x}}{\sigma_x}\)), therefore serves as a moment estimator of \(k\). This direct correspondence between the Weibull parameters and empirical statistics of partial correlations provides a foundation to develop regularization penalties that adapt to the signal in the empirical distribution.

\subsection{Extreme Value Distributions as Regularization Penalties}\label{extreme-value-distributions-as-regularization-penalties}

The Weibull distribution belongs to the broader family of extreme value distributions, which arise from the Fisher-Tippett-Gnedenko theorem representing the extreme value analog of the Central Limit Theorem \citep{fisher1928limiting,gnedenko1943distribution}. This theorem establishes that the limiting distribution of extreme order statistics must belong to one of three families: Gumbel (Type I), Fréchet (Type II), or Weibull (Type III). Their minima form applies to absolute partial correlations as they are bounded below at zero with mass often concentrated near the boundary due to the partialliing of shared covariance. The Weibull distribution unifies the EVT family, connecting each distribution through transformations made explicit in Appendix \ref{appendix-evt} and bridging to the exponential distribution's established use as a regularization penalty for linear regression \citep{breheny2015group,vanderwerken2011variable,wang2018variable}.

\subsubsection{Formal Penalties}\label{formal-penalties}

Regularization penalties can be described in a general form of the penalized GGM log-likelihood:

\[
\hat{\mathbf{K}} = \arg \max_{\mathbf{K} \in \mathcal{S}_{++}^p} 
\left[ \log \det(\mathbf{K}) - \mathrm{tr}(\mathbf{S}\mathbf{K}) - 
\sum_{i < j} p_{\lambda_{ij}}(|k_{ij}|) \right],
\]

\noindent where \(\mathbf{K} = \boldsymbol{\Sigma}^{-1}\) is the \(p \times p\) precision matrix, \(\mathbf{S}\) is the \(p \times p\) sample covariance matrix, \(k_{ij}\) is the \((i, j)\)-element of \(\mathbf{K}\), \(p_{\lambda_{ij}}(\cdot)\) is a penalty function, and \(\lambda_{ij} > 0\) is an element-wise regularization parameter that controls the degree of penalization. The penalty is applied only to off-diagonal elements to preserve positive definiteness of \(\mathbf{K}\).

The EXP (Exponential) penalty, as defined by Wang et al. (\citeyear{wang2018variable}), takes the form:

\[
p(x;\, \lambda, \gamma) = \lambda  \cdot \bigg( 1 - e^{\frac{|x|}{\gamma}} \bigg),
\]

\noindent where \(\gamma\) represents a tunable parameter. When \(\gamma\) is small, the penalty approximates the \(\ell_0\)-norm: \(p(|\theta|) \approx \lambda I (\theta \neq 0)\). The derivative of the penalty takes the from:

\[
p^\prime(x;\, \lambda, \gamma) = \lambda  \cdot \frac{1}{\gamma} e^{-\frac{|x|}{\gamma}},
\]

\noindent which makes the derivative decay exponentially and approach zero asymptotically, but never actually reaching it.

For the Weibull and Gumbel distributions to be theoretically capable of yielding oracle properties, their penalty functions must satisfy three conditions \citep{fan2009network,fan2001variable}.

\begin{enumerate}
\def\labelenumi{\arabic{enumi}.}
\item
  \textit{Concavity on $(0, \infty)$}: the rate of penalization decreases as edge weight magnitude increases, which can be formally expressed as \(p_\lambda^{\prime \prime}(|k_{ij}|) \leq 0\) for \(|k_{ij}| > 0\).
\item
  \textit{Continuous differentiability on $(0, \infty)$}: permits use of the local linear approximation algorithm (LLA) for optimization \citep{fan2009network,zou2008one} and satisfies the regularity conditions required for oracle property proofs \citep{fan2001variable,wang2018variable}.
\item
  \textit{Vanishing derivative}: \(p_\lambda^\prime(|k_{ij}|) \to 0\) as \(|k_{ij}| \to \infty\), ensuring that large edges are penalized negligibly and estimated with vanishing bias, satisfying the unbiasedness condition of Fan and Li (\citeyear{fan2001variable}).
\end{enumerate}

\noindent We formally demonstrate the these three penalties can satisfy all three criteria but the Weibull and Gumbel distributions first require some modifications.

The Weibull distribution satisfies all three conditions when \(k \le 1\) but requires an adjustment when \(k > 1\). When \(k > 1\), the mode of the Weibull distribution is no longer zero and instead takes the value of \(x^* = \gamma \bigg( \frac{k - 1}{k} \bigg)^{\frac{1}{k}}\). The consequence is that the concavity condition is violated because penalization is most aggressive at the mode and monotonically decreasing on either side of it, leading to less aggressive penalization on edges smaller than the mode and undermining the sparsity property.

To address this issue, the Weibull penalty can be defined piecewise (similar to other non-convex penalties; \citealp{zhang2010nearly,fan2001variable}) to ensure concavity for all \(k > 0\). The Weibull PDF, used as the penalty derivative, is capped at its modal value for \(|x| \leq x^*\):

\[
p^{\prime}(x;\, \lambda, \gamma, k) = \lambda \cdot \frac{k}{\gamma} \times
\begin{cases}
\left(\dfrac{x^*}{\gamma}\right)^{k-1} e^{-(x^*/\gamma)^k} & 
\text{if } k > 1 \text{ and } |x| \leq x^* \\[10pt]
\left(\dfrac{|x|}{\gamma}\right)^{k-1} e^{-(|x|/\gamma)^k} & 
\text{otherwise}
\end{cases}.
\]

\noindent The term, \(\frac{k}{\gamma}\left(\frac{x^*}{\gamma}\right)^{k-1}e^{-(x^*/\gamma)^k}\), is denoted \(d^*\), the value of the derivative at the mode, \(x^*\). The resulting derivative is monotonically non-increasing for all \(k > 0\), restoring the concavity condition. The corresponding penalty, given by the Weibull CDF, follows by integration:

\[
p(x;\, \lambda, \gamma, k) = \lambda \times
\begin{cases}
d^* \cdot |x| & 
\text{if } k > 1 \text{ and } |x| \leq x^* \\[10pt]
1 - e^{-(|x|/\gamma)^k} + \delta & 
\text{otherwise}
\end{cases},
\]

\noindent where \(\delta = d^* \cdot x^* - \left(1 - e^{-(x^*/\gamma)^k}\right)\) is a continuity shift ensuring the two pieces meet at \(x^*\). The resulting penalty is concave and continuous for all \(k > 0\).

The standard Gumbel distribution is defined on \((-\infty, \infty)\) and peaks at its location parameter (which is set to zero), meaning its derivative does not vanish as \(|x| \to \infty\) in the negative direction and penalization is applied to negative values. Folding the Gumbel distribution resolves this issue by reflecting the distribution about zero, yielding a symmetric penalty bounded on \([0, \infty)\) that satisfies all three conditions. The folded Gumbel penalty takes the form \citep{nadarajah2015new}:

\[
p(x;\, \lambda, \gamma) = 
\lambda \cdot \bigg[ \bigg( 1 - e^{-e^{\frac{|x|}{\gamma}}} \bigg) - \bigg( 1 - e^{-e^{-\frac{|x|}{\gamma}}} \bigg) \bigg] = 
\lambda \cdot \bigg( e^{-e^{-\frac{|x|}{\gamma}}} - e^{-e^{\frac{|x|}{\gamma}}} \bigg),
\]

\noindent where the left-hand side makes explicit the minima Gumbel motivation and the right-hand side is the simplified form, which is identical to the folded maxima Gumbel. The derivative is equivalently expressed as the folded maxima Gumbel PDF for simplicity:

\[
p^\prime(x;\, \lambda, \gamma) = \frac{\lambda}{\gamma} \left( e^{-\frac{|x|}{\gamma} - e^{-\frac{|x|}{\gamma}}} + e^{\frac{|x|}{\gamma} - e^{\frac{|x|}{\gamma}}} \right).
\]

\noindent With these modifications, both the Weibull and Gumbel penalties satisfy the three conditions that are necessary to achieve oracle properties as a regularization penalty. 

\subsubsection{Oracle Properties}

In this section, we present the theorems and corollaries establishing the oracle properties of these penalties with their formal proofs provided in Appendix \ref{appendix-proof}.

Recall that $\mathbf X_1,\ldots,\mathbf X_n$ are independent
and identically distributed random vectors from $N_p(\mathbf0,\boldsymbol\Sigma_0)$. Define the population
precision matrix as $\mathbf K_0=(k_{ij}^0)_{p\times p}=\boldsymbol\Sigma_0^{-1}$, and denote the true edge set by 
\[
\mathcal A
=\{(i,j):1\leq i<j\leq p,\ k_{ij}^0\neq0\}
.
\]
with cardinality $ s=|\mathcal A|$. Let  $\mathcal A^c$ represent the remaining off-diagonal pairs. Let $\boldsymbol\Omega_0$ and $\widehat{\boldsymbol\Omega}$ denote the
population and estimated partial correlation matrices, respectively, where for any $i \ne j$,
\[
\omega_{ij}^0=-\frac{k_{ij}^0}{\sqrt{k_{ii}^0k_{jj}^0}},
\quad \mathrm{and } \quad
\widehat\omega_{ij}
=-\frac{\widehat k_{ij}}{\sqrt{\widehat k_{ii}\widehat k_{jj}}}, 
\]
where $\omega_{ii}^0 = \widehat{\omega}_{ii} = 1$, and $\widehat{\mathbf{K}} = (\widehat{k}_{ij})_{p\times p}$ is a generic estimator of $\mathbf{K}_0$.

We impose the following assumption on the precision matrix: 
\begin{enumerate}
\item[(i)] There exist constants $0<c<C<\infty$ such that the eigenvalues of $\mathbf K_0$ satisfy $c \le \lambda_{\min}(\mathbf{K}_0) \le \lambda_{\max}(\mathbf{K}_0) \le C$.
\end{enumerate}
Condition (i) is a standard regularity requirement in high-dimensional covariance and precision matrix estimation, ensuring that $\boldsymbol{\Sigma}_0$ and $\mathbf{K}_0$ are well-conditioned \citep{rothman2008sparse,ravikumar2011high,xue2012positive,xue2013minimax}.

Let $k_{\min}=\min_{(i,j)\in\mathcal A}|k_{ij}^0|$, $\tau_n=\sqrt{\frac{\log p}{n}}$, and $r_n=\sqrt{\frac{s\log p}{n}}$. Next, we present the assumptions on the penalty function:
\begin{enumerate}
\item[(ii)] For every $i<j$, $p_{\lambda_{ij}}(0)=0$ and
$p_{\lambda_{ij}}$ is nondecreasing and concave on $[0,\infty)$, continuous
at zero, and continuously differentiable on $(0,\infty)$. 

\item[(iii)] As $n\to\infty$, the penalty derivative satisfies
\begin{equation}
 \frac{\eta_n}{\tau_n}= \frac{1}{\tau_n}\max_{(i,j)\in\mathcal A}
 \sup_{t\geq k_{\min}/2}p_{\lambda_{ij}}^{\prime}(t) \to 0,
\label{eq:first-order-bias}
\end{equation}
and, for every fixed constant $M>0$,
\begin{equation}
\frac{1}{\tau_n}
\min_{(i,j)\in\mathcal A^c}
\inf_{0<t\leq Mr_n}
p_{\lambda_{ij}}^{\prime}(t)
\to +\infty.
\label{eq:null-derivative-common}
\end{equation}
\end{enumerate}

Condition (ii) encompasses a broad class of penalty functions, including the EXP, modified Weibull, and folded Gumbel penalty functions.  In Condition (iii), \eqref{eq:first-order-bias} asymptotically eliminates the first-order shrinkage bias at the true edges, and \eqref{eq:null-derivative-common} ensures that the penalty derivative near the origin dominates the null-edge scores, which have stochastic order $O_{p}(\tau_n)$. 

Let $\|\cdot\|_F$ denote the Frobenius norm. The following theorem establishes the convergence rate of the proposed precision and partial correlation matrix estimators.

\begin{theorem}[Convergence rates]
\label{thm:high-dimensional-rate}
Suppose that 
\(
{s\log p}=o(n)
\),  Assumptions (i)--(iii) hold, and \(\frac{1}{r_n}k_{min}\to\infty\) .
Then, the proposed penalized log-likelihood method has a sequence of local maximizers
$\widehat{\mathbf K}$ satisfying that
\begin{equation}
\|\widehat{\mathbf K}-\mathbf K_0\|_F
=O_p\!\left(\sqrt{\frac{s\log p}{n}}\right).
\label{eq:high-dimensional-rate}
\end{equation}
Moreover, the associated partial correlation matrix estimator $\widehat{\boldsymbol\Omega}$ satisfies
\begin{equation}
\|\widehat{\boldsymbol\Omega}-\boldsymbol\Omega_0\|_F
=O_p\!\left(\sqrt{\frac{s\log p}{n}}\right).
\label{eq:partial-correlation-rate}
\end{equation}
\end{theorem}

Theorem~\ref{thm:high-dimensional-rate} accommodates the high-dimensional setting where both $p$ and $s$ diverge as $n$ grows. The resulting convergence rate in Theorem~\ref{thm:high-dimensional-rate} matches the classical result for Gaussian graphical models (see Corollary~1 of \citealp{rothman2008sparse}) and Gaussian copula graphical models (see Corollary~1 of \citealp{xue2012regularized}).

We next establish the oracle property of the proposed estimator. Define the oracle parameter space as $$\mathcal M_{\mathcal A}
=
\{\mathbf K\succ0:k_{ij}=0
\text{ whenever }(i,j)\notin\mathcal A,\ i<j\}.$$
For any $\mathbf K\in\mathcal M_{\mathcal A}$, let
$\boldsymbol\vartheta_{\mathcal A}(\mathbf K)
=
\left(
k_{11},\ldots,k_{pp},
(k_{ij})_{(i,j)\in\mathcal A}
\right)^{\mathsf T}$ collect the non-zero elements, and write
$\boldsymbol\vartheta_{\mathcal A,0}
=
\boldsymbol\vartheta_{\mathcal A}(\mathbf K_0)$.
Let $\mathbf K(\boldsymbol\vartheta_{\mathcal A})$ denote the symmetric parameterized by $\boldsymbol\vartheta_{\mathcal A}$. Under this oracle model, the density of a single observation $\mathbf x\in\mathbb{R}^p$ is
\[
f_{\mathcal A}(\mathbf x;\boldsymbol\vartheta_{\mathcal A})
=
(2\pi)^{-p/2}
\det\{\mathbf K(\boldsymbol\vartheta_{\mathcal A})\}^{1/2}
\exp\!\left\{
-\frac12\mathbf x^{\mathsf T}
\mathbf K(\boldsymbol\vartheta_{\mathcal A})\mathbf x
\right\}.
\]
The corresponding score vector and per-observation Fisher information matrix are defined as
\[
\boldsymbol S_{\mathcal A}(\mathbf X)
=
\left.
\frac{\partial}
{\partial\boldsymbol\vartheta_{\mathcal A}}
\log f_{\mathcal A}
(\mathbf X;\boldsymbol\vartheta_{\mathcal A})
\right|_{\boldsymbol\vartheta_{\mathcal A}
=\boldsymbol\vartheta_{\mathcal A,0}},
\quad \mathrm{and } \quad
\mathcal J_{\mathcal A,n}
=
\mathbb{E}_0\!\left[
\boldsymbol S_{\mathcal A}(\mathbf X)
\boldsymbol S_{\mathcal A}(\mathbf X)^{\mathsf T}
\right],
\]
where  $\mathbb{E}_0$ denotes the expectation under the true distribution. Let $\boldsymbol g_{\mathcal A}(\boldsymbol\vartheta_{\mathcal A})
=
\left(
-\frac{k_{ij}}{\sqrt{k_{ii}k_{jj}}}
\right)_{(i,j)\in\mathcal A}$ collect the active partial correlations, and let $\mathbf G_{\mathcal A,n}
=
\left.
\frac{\partial
\boldsymbol g_{\mathcal A}(\boldsymbol\vartheta_{\mathcal A})}
{\partial\boldsymbol\vartheta_{\mathcal A}^{\mathsf T}}
\right|_{\boldsymbol\vartheta_{\mathcal A}
=\boldsymbol\vartheta_{\mathcal A,0}}$ denote its Jacobian matrix evaluated at $\boldsymbol\vartheta_{\mathcal A,0}$. Assuming the limit exists, the asymptotic oracle  information matrix
$\mathcal I_{\mathcal A}^{\,o}$ is defined by
\begin{equation}
(\mathcal I_{\mathcal A}^{\,o})^{-1}
=
\lim_{n\to\infty}
\mathbf G_{\mathcal A,n}
\mathcal J_{\mathcal A,n}^{-1}
\mathbf G_{\mathcal A,n}^{\mathsf T}.
\label{eq:oracle-information}
\end{equation}

\begin{theorem}[Oracle property]
\label{thm:oracle-fixed-p}
Suppose that $p$ and $s$ are fixed,
 Assumptions (i)--(iii) hold, $\sqrt{n}k_{\min}\to+\infty$, and the limiting matrix $(\mathcal I_{\mathcal A}^{\,o})^{-1}$ in \eqref{eq:oracle-information} exists and is positive definite. Then 
as $n\to\infty$, there exists a sequence of local maximizers
$\widehat{\mathbf K}$ of the proposed penalized log-likelihood method that achieves the graphical model selection consistency:
\begin{equation}
P\!\left(\{(i,j):i<j,\ \widehat k_{ij}\neq0\}=\mathcal A\right)
\to 1.
\label{eq:oracle-selection}
\end{equation} 
Moreover, as $n\to\infty$, we have the following asymptotic normality results:
\begin{equation}
\sqrt n\left\{
(\widehat\omega_{ij})_{(i,j)\in\mathcal A}
-(\omega_{ij}^0)_{(i,j)\in\mathcal A}
\right\}
\longrightarrow
N_s\!\left(\mathbf0,(\mathcal I_{\mathcal A}^{\,o})^{-1}\right).
\label{eq:oracle-asymptotic-normality}
\end{equation}
\end{theorem}

Theorem \ref{thm:oracle-fixed-p} shows that the proposed estimator possesses the oracle property in the sense of \citet{fan2001variable}: with probability approaching one, it correctly identifies the true underlying graph structure and estimates the nonzero partial correlations as efficiently as if the true support were known in advance. The condition $\sqrt{n}k_{\min} \to \infty$ is a standard minimal signal strength condition in the literature \citep{fan2009network,fan2014strong,ravikumar2011high,zhang2026copula}.

The following three corollaries provide explicit rate conditions on the tuning parameters under which the EXP, modified Weibull, and folded Gumbel penalties satisfy the oracle property. Throughout, we consider fixed $p,s$ and set a common penalty parameter as $\lambda_{ij}= \lambda_n$.

\begin{corollary}[EXP penalty]
\label{cor:oracle-exp}
Suppose that $p$ and $s$ are fixed,
 Assumption (i) hold, $\sqrt{n}k_{\min}\to+\infty$, and the limiting matrix $(\mathcal I_{\mathcal A}^{\,o})^{-1}$ in \eqref{eq:oracle-information} exists and is positive definite. When the penalty is given by $p(x;\lambda_n,\gamma_n)=\lambda_n\{1-\exp(-|x|/\gamma_n)\}$, if 
\begin{equation}
\lambda_n\to0,\quad\gamma_n\to0,\quad\sqrt n\,\gamma_n\to\infty,\quad
\frac{\sqrt n\,\lambda_n}{\gamma_n}\to\infty,\quad\mathrm{and}\quad
\frac{\sqrt n\,\lambda_n}{\gamma_n}
 e^{-k_{\min}/(2\gamma_n)}\to0,
\label{eq:oracle-exp-tuning}
\end{equation}
then the conclusions of Theorem~\ref{thm:oracle-fixed-p} hold.
\end{corollary}

For the EXP penalty, the ratio $\lambda_n/\gamma_n$ controls thresholding near the origin to eliminate false edges, while the decaying 
exponential term asymptotically eliminates shrinkage bias at true edges. The special case
$k=1$ in the modified Weibull family corresponds to this formulation.

\begin{corollary}[Modified Weibull penalty]
\label{cor:oracle-weibull}
Suppose that $p$ and $s$ are fixed,
 Assumption (i) hold, $\frac{1}{\gamma_n}k_{\min}\to+\infty$, and the limiting matrix $(\mathcal I_{\mathcal A}^{\,o})^{-1}$ in \eqref{eq:oracle-information} exists and is positive definite. When the modified Weibull penalty is employed with a fixed $k>0$, if
\begin{equation}
\lambda_n\to0,\ \gamma_n\to0,\ \sqrt n\,\gamma_n\to\infty,\ 
\frac{\sqrt n\,\lambda_n}{\gamma_n}\to\infty,\ \mathrm{and}\ 
\frac{\sqrt n\,\lambda_n k}{\gamma_n}
\left(\frac{k_{\min}}{2\gamma_n}\right)^{k-1}
\exp\!\left[-\left(\frac{k_{\min}}{2\gamma_n}\right)^k\right]
\to0,
\label{eq:oracle-weibull-tuning}
\end{equation}
then the conclusions of Theorem~\ref{thm:oracle-fixed-p} hold.
\end{corollary}

For $0<k\leq1$, the penalty derivative is decreasing on $(0,\infty)$. For $k>1$, the
modification caps its increasing part; since $k_{\min}$ is bounded away from zero and
$\gamma_n\to0$, every true edge eventually falls in the decreasing region.

\begin{corollary}[Folded Gumbel penalty]
\label{cor:oracle-gumbel}
Suppose that $p$ and $s$ are fixed,
 Assumption (i) hold, $\sqrt{n}k_{\min}\to+\infty$, and the limiting matrix $(\mathcal I_{\mathcal A}^{\,o})^{-1}$ in \eqref{eq:oracle-information} exists and is positive definite. When the folded Gumbel penalty is employed, if
\begin{equation}
\lambda_n\to0,\quad\gamma_n\to0,\quad\sqrt n\,\gamma_n\to\infty,\quad
\frac{\sqrt n\,\lambda_n}{\gamma_n}\to\infty,\quad\mathrm{and}\quad
\frac{\sqrt n\,\lambda_n}{\gamma_n}e^{-k_{\min}/(2\gamma_n)}\to0,
\label{eq:oracle-gumbel-tuning}
\end{equation}
then the conclusions of Theorem~\ref{thm:oracle-fixed-p} hold.
\end{corollary}

Because the folded Gumbel penalty derivative is positive, decreasing and bounded from above by an exponential decay term, it yields the same sufficient rate conditions as the EXP penalty.

The rate conditions in \eqref{eq:oracle-exp-tuning}--\eqref{eq:oracle-weibull-tuning} are compatible and easily satisfied. For example, When $k_{\min}$ is bounded away from zero, the deterministic sequences 
$\gamma_n=n^{-1/4}$ and $\lambda_n=n^{-1/2}$ the conditions of all three corollaries for any fixed $k>0$. When the tuning parameters are selected adaptively from the
data, the same conclusions follow provided that the corresponding stochastic rates hold in probability.

\subsubsection{Data-Adaptive Penalties}\label{data-adaptive-penalties}

By grounding regularization penalties in distributions, their empirical distribution statistics can be used to flexibly adapt to the signal from the empirical data. The MLE parameters of each penalty can be estimated from the empirical (absolute) partial correlations, \(x\), which will adjust how close they approximate the \(\ell_0\) penalty. For the EXP and Gumbel penalties, the scale (\(\gamma\)) parameter can be estimated directly. For the Exponential distribution, the MLE \(\gamma\) is defined by the mean, \(\hat{\gamma}_{MLE} = \bar{x}\) \citep{johnson1994continuous}. For the Weibull distribution, the MLE for the shape parameter, \(k\), must be solved numerically using,

\[
\frac{1}{\hat{k}} + \frac{1}{n} \sum_{i=1}^n \ln x_i - \frac{\sum_{i=1}^n x_i^{k} \ln x_i}{\sum_{i=1}^n x_i^{k}} = 0,
\]

\noindent where after the MLE \(\gamma\) can be estimated analytically using \(k\): \(\hat{\gamma}_{MLE} = \left( \frac{1}{n} \sum_{i=1}^n x_i^{\hat{k}} \right)^{\frac{1}{\hat{k}}}\) \citep{rinne2008weibull}.

The MLE for folded Gumbel's \(\gamma\) needs to be derived from the standard Gumbel log-likelihood \citep{kotz2000extreme},

\[
{\cal{L}}(\gamma) = \sum_{i=1}^n \left[ \frac{x_i}{\gamma} - e^{x_i/\gamma} - \log\gamma \right],
\]

\noindent where the log-likelihood must be folded following Johnson et al. (\citeyear{johnson1994continuous}), \(f_{fold}(x) = f(x) + f(-x)\), leading to \citep{nadarajah2015new}:

\[
{\cal{L}}(\gamma) = \sum_{i=1}^n \left[ \log\left( e^{\frac{x_i}{\gamma} - e^{\frac{x_i}{\gamma}}} + e^{-\frac{x_i}{\gamma} - e^{-\frac{x_i}{\gamma}}} \right) - \log\gamma \right].
\]

The MLE parameters characterize the central tendency of the full empirical partial correlation distribution, which spans both signal and noise. Setting \(\gamma\) directly to its MLE estimate would therefore regularize signal and noise, undermining the penalty's ability to selectively shrink near-zero edges. Instead, an adjustment to the MLE \(\gamma\), \(\gamma^*\), is needed to place each penalty's mass over the noise floor of the distribution.

To identify an appropriate \(\gamma^*\), we use the EXP penalty as a guiding case: Wang et al. (\citeyear{wang2018variable}) established through simulation that \(\gamma = 0.01\) recovers oracle properties, providing a validated operating point \citep{GGMncv}. Rather than treating this as an arbitrary fixed value, we calibrated the quantile of each distribution to correspond to \(\gamma^* \approx 0.01\) given ``typical'' partial correlation distributions. We used the median MLE scale (0.106) and shape (1.091) estimates as our ``typical'' guides, in effect, using Wang et al.'s result to calibrate the quantile-based selection rule that then generalizes to all three penalties.

To this extent, we identified the approximate quantile that aligned with \(\gamma^* \approx 0.01\) for all three penalties that were based in the empirical data. Using nearest rounded parameters, 0.10 for all three penalties' initial \(\gamma\) and 1.00 for Weibull's \(k\), we estimated the nearest quantile that aligned with \(\gamma^* = 0.01\). By using \(k = 1\), the Weibull penalty collapses to the EXP penalty. The resulting quantile was roughly the first decile (\(p = 0.10\)) for the EXP and Weibull penalties, \(\gamma^* = \gamma (-\ln(1 - p))\), with a value of about 0.0105.

Because the Gumbel penalty is represented by the folded Gumbel distribution, the quantile must be numerically derived as the algebraic manipulation becomes intractable, \(e^{-e^{-\frac{q}{\gamma}}} - e^{-e^{\frac{q}{\gamma}}} = p\). When solving for \(p = 0.10\) (first decile) and \(\gamma = 0.10\), then \(q\) and, consequently \(\gamma^*_{Gumbel}\), was about 0.0136.

We note that this calibration substitutes one free parameter (\(\gamma\)) for another (first decile of the distribution); however, unlike a fixed \(\gamma\), the resulting \(\gamma^*\) shifts with the empirical distribution, scaling the penalty's aggressiveness to the noise floor of each dataset rather than assuming a universal signal level.\footnote{Sensitivity analyses evaluating the 5\textsuperscript{th} and 15\textsuperscript{th} percentile as alternative thresholds for EXP (\href{https://osf.io/6hxkt/files/yuh8w}{\color{blue} OSF}), Gumbel (\href{https://osf.io/pa2mq}{\color{blue} OSF}), and Weibull (\href{https://osf.io/p8bwm}{\color{blue} OSF}) are provided as Supplemental Materials \citep{christensen2026adaptive_osf}. Across both network topologies, the quantile threshold governed a monotone sensitivity/specificity trade-off, with differences becoming more apparent at \(N\) = 10,000. All adaptive penalties continued to accumulate sensitivity with increasing sample size.} Together, the three penalties---EXP, Weibull, and Gumbel---form a coherent data-adaptive regularization framework unified by extreme value theory.

\section{Present Research}\label{present-research}

The present study evaluates whether grounding regularization penalties in the empirical distribution of partial correlations improves GGM estimation relative to contemporary standards in psychology. Our proposed family of data-adaptive penalties are evaluated against EBICglasso, BGGM, and several non-convex penalties including Smoothly Clipped Absolute Deviation (SCAD; \citealp{fan2001variable}), Minimax Concave Penalty (MCP; \citealp{zhang2010nearly}), Atan \citep{wang2016variable}, and EXP (\citealp{wang2018variable}; see Appendix \ref{appendix-penalties}). EBICglasso and BGGM represent field standards, SCAD and MCP establish whether non-convex penalties more broadly outperform \(\ell_1\) regularization \citep{fan2009network,williams2020beyond}, and  Atan and EXP are static \(\ell_0\) approximating non-convex penalties with EXP directly isolating whether adaptivity contributes anything beyond \(\ell_0\) approximation alone.

These penalties are evaluated using simulated networks that vary in topology: Stochastic Block Model (SBM; \citealp{holland1983stochastic}) and Small-world Network (SWN; \citealp{watts1998collective}). These topologies represent divergent network structures expected to be common in psychology: SBMs consist of densely connected sets of nodes forming clusters that represent multidimensional constructs \citep{christensen2026network,sekulovski2025stochastic} whereas SWNs exhibit local clustering without underlying dimensionality, representing a unified system. Performance was assessed in terms of edge recovery, parameter bias (edge and centrality), and rank-order centrality congruence. Together, these comparisons are designed to disentangle whether improvements in edge and centrality recovery stem from non-convexity, \(\ell_0\) approximation, or adaptivity across a broad range of plausible psychological network structures. Finally, to illustrate the practical consequences, we present an empirical example using
Johnson's (\citeyear{johnson2014measuring}) IPIP-NEO \citep{goldberg2006international} dataset (Appendix \ref{appendix-empirical}), examining a cross-sectional comparison of methods at a typical sample size and model selection consistency across increasing subsample sizes.

\section{Method}\label{method}

Multivariate normal continuous data were generated from two network topologies: SBM and SWN. Following Christensen and Choi (\citeyear{christensen2026network}), we re-analyzed the 290 empirical datasets to determine the appropriate parameter space for each simulated topology. Sample size was evaluated across contexts from common applications in psychology, 100, 250, 500, and 1,000 \citep{huth2025statistical} to large samples of 2,500 and 10,000. The smaller samples served the purpose of providing practical recommendations to applied researchers whereas the larger samples, particularly 10,000, were intended to numerically investigate the oracle properties of the network estimation methods.

\subsection{Simulation Design}\label{simulation-design}

The SBM topology varied number of communities (2, 3, 4) and variables per community (4, 6, 8), selected to cover the majority of total variables across psychological networks (72.8\% of the 290 networks; \citealp{christensen2026network}). Within- (0.80, 0.90) and between-community (0.20, 0.40) edge probabilities mirrored the edge density patterns of EBICglasso- and BGGM-estimated networks from the empirical datasets.

The SWN topology used 10, 20, 30, and 40 variables (covering 68.6\% of the 290 networks) and densities of 0.25, 0.50, and 0.75, consistent with the majority of networks in the literature \citep{christensen2026network,wysocki2019penalty}. Rewiring probabilities (\(r_p\) = 0.10, 0.21, 0.32) were determined by computing the small-worldness metric, \(\omega\) \citep{telesford2011ubiquity}, for EBICglasso- and BGGM-estimated networks across the empirical datasets (\href{https://osf.io/6hxkt/files/h25sq}{\color{blue} OSF}), using two standard deviations above and below the mean from Figure 2A in Telesford et al. (\citeyear{telesford2011ubiquity}).

For both topologies, SNR for the edge weights was set to 0.65 (weak), 1 (moderate), and 1.35 (strong), representing \(\pm2 SD\) in the empirical partial correlations (\href{https://osf.io/6hxkt/files/629du}{\color{blue} OSF}). Both topologies were fully factorial with the SBM conditions of sample size \(\times\) number of communities \(\times\) number of variables per community \(\times\) within-community edge probability \(\times\) between-community probability \(\times\) SNR (6 \(\times\) 3 \(\times\) 3 \(\times\) 2 \(\times\) 2 \(\times\) 3) and SWN conditions of sample size \(\times\) number of variables \(\times\) network density \(\times\) rewiring probability \(\times\) SNR (6 \(\times\) 4 \(\times\) 3 \(\times\) 3 \(\times\) 3). For each condition across topologies, 100 replicate networks were generated. The full design was therefore 648 conditions and 64,800 replicates for each topology, resulting in 1,296 total conditions with 129,600 total replicates.

\subsection{Network Generation}\label{network-generation}

Full procedural details for both network topologies are provided in the Supplemental Materials (\href{https://osf.io/6hxkt/files/37srg}{\color{blue} OSF}), which describe five general steps: adjacency matrix generation, edge weight generation, edge weight assignment, negative sign assignment, and, if necessary, conditioning for positive definiteness.

The SBM adjacency matrices were generated as Planted Partition Models \citep{condon2001algorithms} with probabilistic edge assignment governed by within- and between-community probabilities. SWN adjacency matrices were generated via a novel procedure extending the Watts-Strogatz rewiring approach \citep{watts1998collective}: an initial ring lattice was pruned to the target density and edges were rewired with a preference for higher-degree nodes to encourage hub-like behavior and degree heterogeneity.

Following Christensen and Choi (\citeyear{christensen2026network}), edge weights were generated to mirror real-world data while preserving the relationship between the partial correlation distribution and sample characteristics (number of variables and sample size). The MLE shape and scale parameters from Huth and colleagues' (\citeyear{huth2025statistical}) datasets were modeled using a seemingly unrelated regression (SUR; \citealp{zellner1962efficient}) to account for their residual correlation (\(r = 0.262\)). The final model explained substantial variance in both shape (\(R^2 = 0.887\)) and scale (\(R^2 = 0.885\)). Condition-specific parameters predicted shape and scale, and bootstrapped SUR residuals introduced variability while preserving their correlation. For each replicate network, \(\frac{p(p - 1)}{2}\) candidate weights were drawn from the Weibull distribution using these predicted parameters, and a subset matching the number of non-zero edges was selected via weighted probability favoring larger weights.

Edge weights were then assigned to match topology-specific structural constraints: for SBM, larger weights were generally assigned to within-community edges, with a hyperparameter allowing some large weights to diffuse to between-community edges; for SWN, weights were assigned inversely proportional to topological distance from the ring lattice \citep{muldoon2016small}. Negative signs were assigned to between-community edges only (SBM) or by flipping variable signs (SWN), with the proportion drawn from \(N(0.343, 0.086)\) based on empirical prevalence.

Both topology procedures produce a sparse partial correlation matrix (i.e., a network), which is then converted to a zero-order correlation matrix to generate data. This matrix is not guaranteed to be positive definite, a well-known challenge in generating sparse partial correlation matrices \citep{joe2006generating,warton2008penalized}; when it was not, ridge conditioning was applied \citep{peeters2020spectral,warton2008penalized}, targeting a spectral condition number of 30 \citep{won2013condition}.

\subsection{Sample Data Generation}\label{sample-data-generation}

Using the zero-order correlation matrix, \(\mathbf{R}\), implied by the sparse network, Cholesky decomposition was performed

\[
\mathbf{R} = \mathbf{U}^\prime \mathbf{U}
\]

\noindent and the sample data matrix of continuous variables was generated

\[
\mathbf{X} = \mathbf{Z} \mathbf{U}
\]

\noindent where \(\mathbf{Z}\) is a \(n \times p\) matrix of random multivariate normal data. Skew was generated using the sinh-arcsinh transformation \citep{jones2009sinh}, which applies \(\sinh(\delta \times (\text{asinh}(z) + \epsilon))\) to standard normal variables where \(\epsilon\) controls skew and \(\delta\) controls kurtosis.

Skew for each variable was allowed to vary randomly with values drawn from -1 to 0 in increments of 0.05. We did not manipulate skew under the premise that small-to-moderate skew is common, and we wanted to maintain empirical realism without adding another factorial condition to an already large simulation.

\subsection{Network Estimation}\label{network-estimation}

\subsubsection{EBICglasso}\label{ebicglasso}

The EBICglasso was estimated using field defaults, where a grid search over the penalty hyperparameter, \(\lambda\), was conducted across 100 candidate values and selected via EBIC with its own hyperparameter controlling preference for model complexity set to 0.50. Following standard procedure \citep{epskamp2018tutorial}, the search bounds were defined as

\[
\lambda_{\max} = \max_{i \neq j} \lvert \hat{r}_{ij} \rvert
\]

\noindent and

\[
\lambda_{\min} = \lambda_{\max} \cdot 0.01,
\]

\noindent where \(\hat{r}_{ij}\) denotes the off-diagonal elements of the sample correlation matrix \(\hat{\mathbf{R}}\). The full sequence of \(\lambda\) values was obtained by exponentiating 100 equally spaced values between \(\log \lambda_{min}\) and \(\log \lambda_{max}\), yielding a geometrically spaced grid from \(\log \lambda_{min}\) to \(\log \lambda_{max}\).

\subsubsection{BGGM}\label{bggm}

Rather than applying a regularization penalty, BGGM estimates a network through Bayesian inference, updating prior beliefs about the network structure, \({\cal{S}}\), and partial correlation parameters, \(\boldsymbol{\Theta}\), to a posterior distribution after observing the data \citep{huth2024simplifying}:

\[
p(\boldsymbol{\Theta}, {\cal{S}} \mid \textrm{data}) \propto 
p(\textrm{data} \mid \boldsymbol{\Theta}, {\cal{S}}) \times p(\boldsymbol{\Theta} \mid {\cal{S}}) 
\times p({\cal{S}}).
\]

We used default settings in the \{easybgm\} package (version 0.4.0; \cite{easybgm}), which calls \{BGGM\} (version 2.1.6; \citealp{williams2020bggm}) for continuous data, placing a matrix-F prior on the covariance matrix with standard deviation 0.25 \citep{mulder2018matrix} and drawing 10,000 posterior samples via MCMC. Edge inclusion was determined using the median probability model \citep{barbieri2004optimal,barbieri2021median}, retaining edges with \(\textrm{BF}_{10} > 1\) (i.e., posterior inclusion probability greater than 0.50) with posterior parameter estimates as the point estimates.

\subsubsection{Non-convex Penalties}\label{non-convex-penalties}

The data-adaptive and non-convex penalties, Atan \citep{wang2016variable}, EXP \citep{wang2018variable}, MCP \citep{zhang2010nearly}, and SCAD (\citealp{fan2001variable}; Table \ref{appendix-penalties}), all followed the same \(\lambda\) selection procedure as EBICglasso except 50 lambda values rather than 100 were used in the search and model selection was based on BIC rather than EBIC. This procedure follows previous applications of these methods in the literature \citep{williams2020beyond}.

Because these penalties render the penalized likelihood non-convex, direct optimization via GLASSO is not possible. Each non-convex penalty, including the adaptive penalties, was therefore optimized via the Local Linear Approximation (LLA; \citealp{zou2008one,fan2014strong}), which converts the problem into a sequence of weighted GLASSO problems by locally approximating the non-convex penalty with a weighted \(\ell_1\) penalty whose weights are the penalty derivative evaluated at the current precision matrix estimate. This process iterates until the maximum absolute change in the precision matrix fell below \(10^{-3}\) or 10,000 iterations were reached.

\subsection{Evaluation Metrics}\label{evaluation-metrics}

Performance was evaluated for edge recovery, parameter bias (edge and centrality), and rank-order centrality congruence (see Appendix \ref{appendix-metrics} for formal definitions). Edge recovery was assessed via sensitivity and specificity. Parameter bias was quantified as mean bias error for true positive edges (MBE\(_{TP}\)) and node strength (MBE\(_{NS}\)), where positive and negative values indicate overestimation and underestimation, respectively. Node strength (the sum of absolute edge weights across a node's connections) served as the centrality measure given its prevalence in the literature. Rank-order centrality congruence was assessed using Kendall's \(\tau_b\) between estimated and population node strength, with qualitative benchmarks of adequate (\(\tau_b \geq 0.50\)), good (\(\tau_b \geq 0.60\)), and excellent (\(\tau_b \geq 0.70\)) following Christensen and Choi (\citeyear{christensen2026network}).

\subsection{Transparency and Openness}\label{transparency-and-openness}

All manuscript preparation, simulations, analyses, and visualizations were carried out using R (version 4.5.1; \citealp{R-base}). Details of the R packages used in this study are presented in Appendix \ref{appendix-code}.

\section{Results}\label{results}

Edge recovery across the SBM and SWN conditions was evaluated to characterize each method's general tendencies (Figure \ref{fig:sen_spec_simulation}). All methods improved in both sensitivity and specificity as sample size increased, with one exception: EBICglasso's sensitivity approached one, but its specificity remained flat or declined, reflecting the tendency of \(\ell_1\) penalization to over-select edges.

\begin{figure}[H]
\centering
\includegraphics[width=\textwidth]{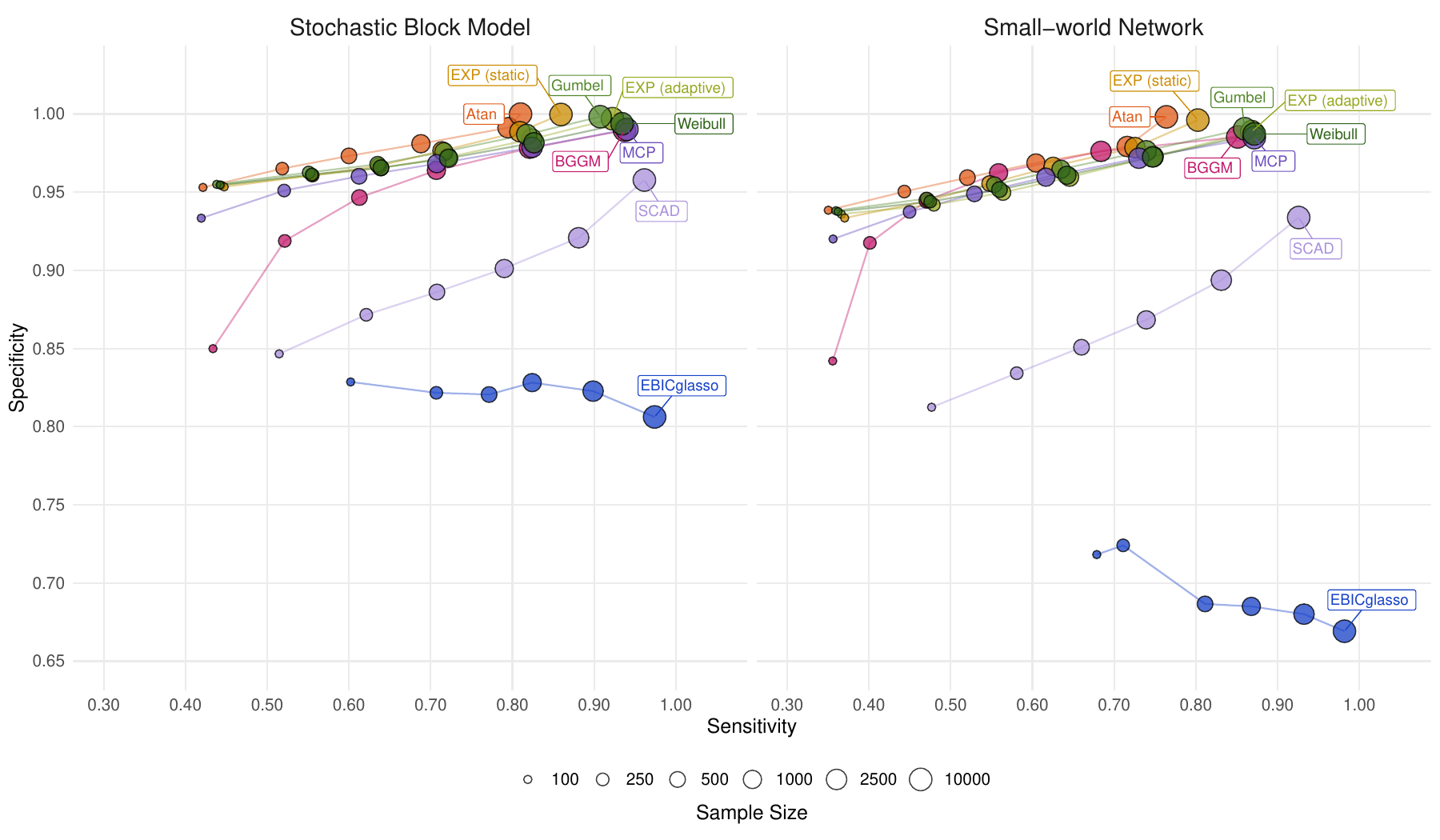}
\caption{\label{fig:sen_spec_simulation} Sensitivity and specificity of edge recovery across methods and sample sizes for the stochastic block model (SBM; left) and small-world network (SWN; right) generating mechanisms. Each point represents the average sensitivity and specificity at a given sample size (100--10,000), with trajectories tracing each method's path across increasing sample sizes. Better overall performance corresponds to the upper-right region of the plot.}
\end{figure}

BGGM showed a conservative pattern---moderate specificity paired with lower sensitivity at smaller sample sizes (N = 100, 250, 500)---consistent with the conservative thresholding behavior of credible interval-based inference, gradually converging with the remaining methods (excluding EBICglasso and SCAD) at the largest sample size tested (N = 10,000). SCAD had the second-highest sensitivity across sample sizes while approaching the near-ceiling specificity of the remaining methods.

The remaining methods (Atan, both EXP variants, Gumbel, MCP, and Weibull) maintained high specificity (\(\geq 0.90\)) with sensitivity increasing as sample size grew. MCP lagged in both metrics until N = 2,500. Atan and EXP (static) tracked one another closely, achieving the highest specificity across sample sizes but leveling off near 0.80 sensitivity at N = 10,000, marking a plateau in recovery. This plateau is where EXP (static) and EXP (adaptive) diverged most: the adaptive penalty continued accruing sensitivity where the static form stalled. The adaptive penalties overall matched Atan and EXP (static) in specificity while exceeding every other method's sensitivity except EBICglasso and SCAD, placing them consistently near the top of all methods on both metrics, especially at larger sample sizes.

\begin{figure}[H]
\centering
\includegraphics[width=\textwidth]{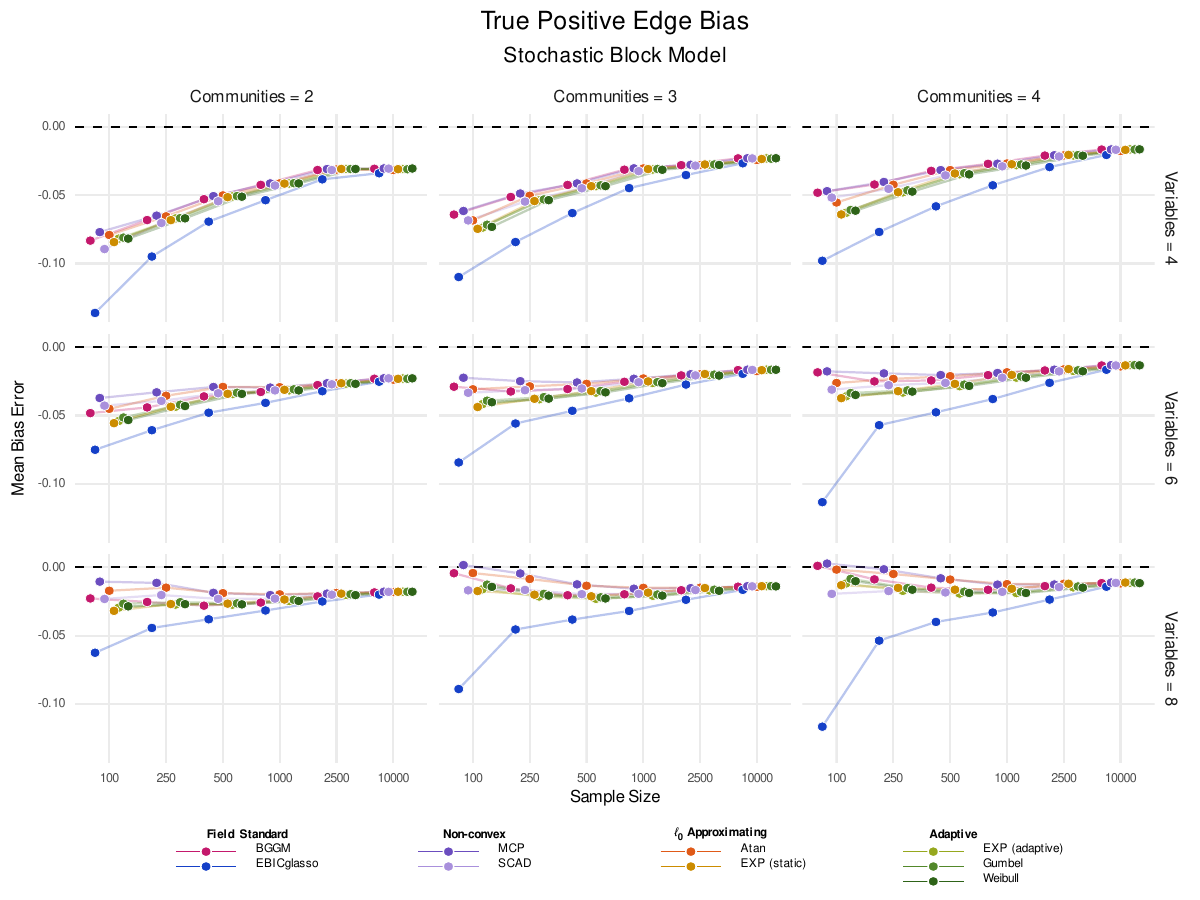}
\caption{\label{fig:sbm_mbe} Mean bias error for true positive edges across methods, sample sizes, and network conditions under the stochastic block model (SBM) generating mechanism and broken down by number of communities (columns) and number of variables per community (rows). Values below zero indicate underestimation and values above zero indicate overestimation of true edge magnitudes in the population network. Each point represents the average MBE at a given sample size.}
\end{figure}

Parameter bias (edge and centrality) and rank-order centrality congruence results below are broken down by sample size, number of communities, and number of variables per community for SBM networks; parallel SWN results, broken down by sample size, network density, and total number of variables, are provided as Supplementary Material (\href{https://osf.io/6hxkt/files/37srg}{\color{blue} OSF}) and mirrored the SBM pattern.

Nearly all methods showed neutral-to-negative edge bias for true positives (TPs), indicating that estimated magnitudes tended to fall below population values (Figure \ref{fig:sbm_mbe}), and bias generally decreased as sample size, number of communities, and number of variables per community increased. Methods clustered into three groups: low bias (Atan, BGGM, MCP, and SCAD), moderate bias (EXP static and adaptive, Gumbel, and Weibull), and comparatively large bias (EBICglasso), with the low--moderate separation smaller than the moderate--large separation. EBICglasso's bias was most pronounced at small sample sizes and converged toward the other methods as N grew.

Consistent with the theoretical behavior of the \(\ell_1\) penalty, median absolute bias for false positive (FP) edges was lowest for EBICglasso (\href{https://osf.io/6hxkt/files/pdf2k}{\color{blue} OSF})---it shrinks small spurious edges toward zero even when it fails to exclude them outright---followed by SCAD, which partially inherits this property from its initial \(\ell_1\) shape while maintaining lower TP bias. The remaining methods were largely comparable, with FP bias decreasing as sample size and network sparsity (fewer communities, fewer variables per community) increased.

These three sources of error---underestimated TP edges, included FP edges, and excluded true edges (\href{https://osf.io/6hxkt/files/pdf2k}{\color{blue} OSF})---jointly propagate into centrality bias (Figure \ref{fig:sbm_centrality_mbe}), with the largest divergence at the smallest sample size (\(N\) = 100): BGGM overestimated centrality while EBICglasso underestimated it. These opposing directions are mechanistically distinct---BGGM's overestimation reflects accumulating small, non-negligible FP edges across many node connections, while EBICglasso's underestimation reflects its heavy penalization of TP edges even as FP edges are minimized. All other methods slightly underestimated centrality, consistent with the negative TP edge bias observed in Figure \ref{fig:sbm_mbe}; even modest edge bias can compound into substantial centrality bias as it accumulates across nodes.

\begin{figure}[H]
\centering
\includegraphics[width=\textwidth]{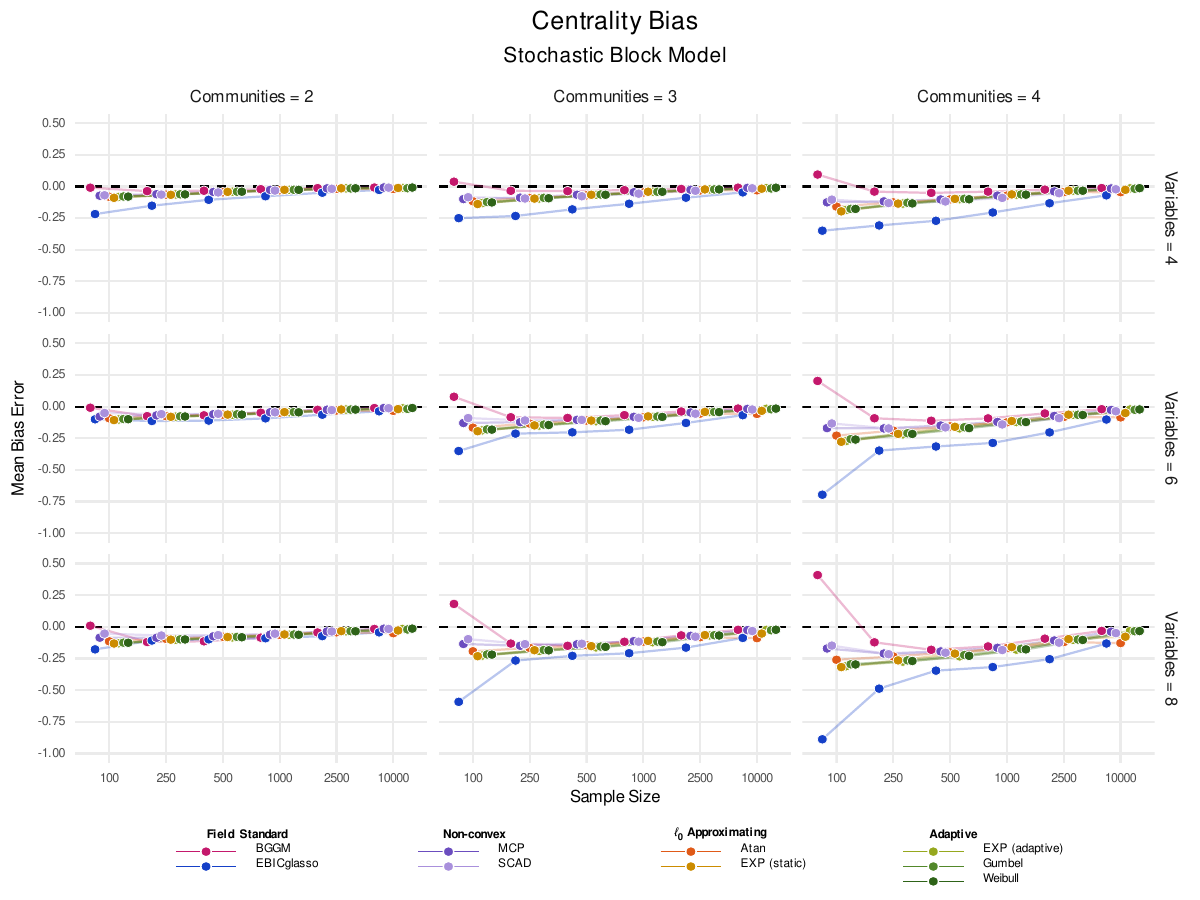}
\caption{\label{fig:sbm_centrality_mbe} Mean bias error for node strength across methods, sample sizes, and network conditions under the stochastic block model (SBM) generating mechanism and broken down by number of communities (columns) and number of variables per community (rows). Values below zero indicate underestimation and values above zero indicate overestimation of true node strength in the population network. Each point represents the average MBE at a given sample size.}
\end{figure}

Centrality bias decreased as sample size increased and total number of variables (communities and variables per community) decreased---opposite in direction to the TP edge bias pattern (Figure \ref{fig:sbm_mbe})---reflecting the compounding role of FPs and FNs in centrality estimation: node strength accumulates errors across \(p - 1\) potential edges, where FP edges contribute upward bias (overestimation) and FN edges contribute downward bias (underestimation), with net centrality bias set by their relative magnitudes across the full edge set (Figure \ref{fig:sen_spec_simulation}). This accumulation amplifies any imbalance between these error sources in larger networks, producing greater centrality bias even as TP bias itself decreases.

\begin{figure}[H]
\centering
\includegraphics[width=\textwidth]{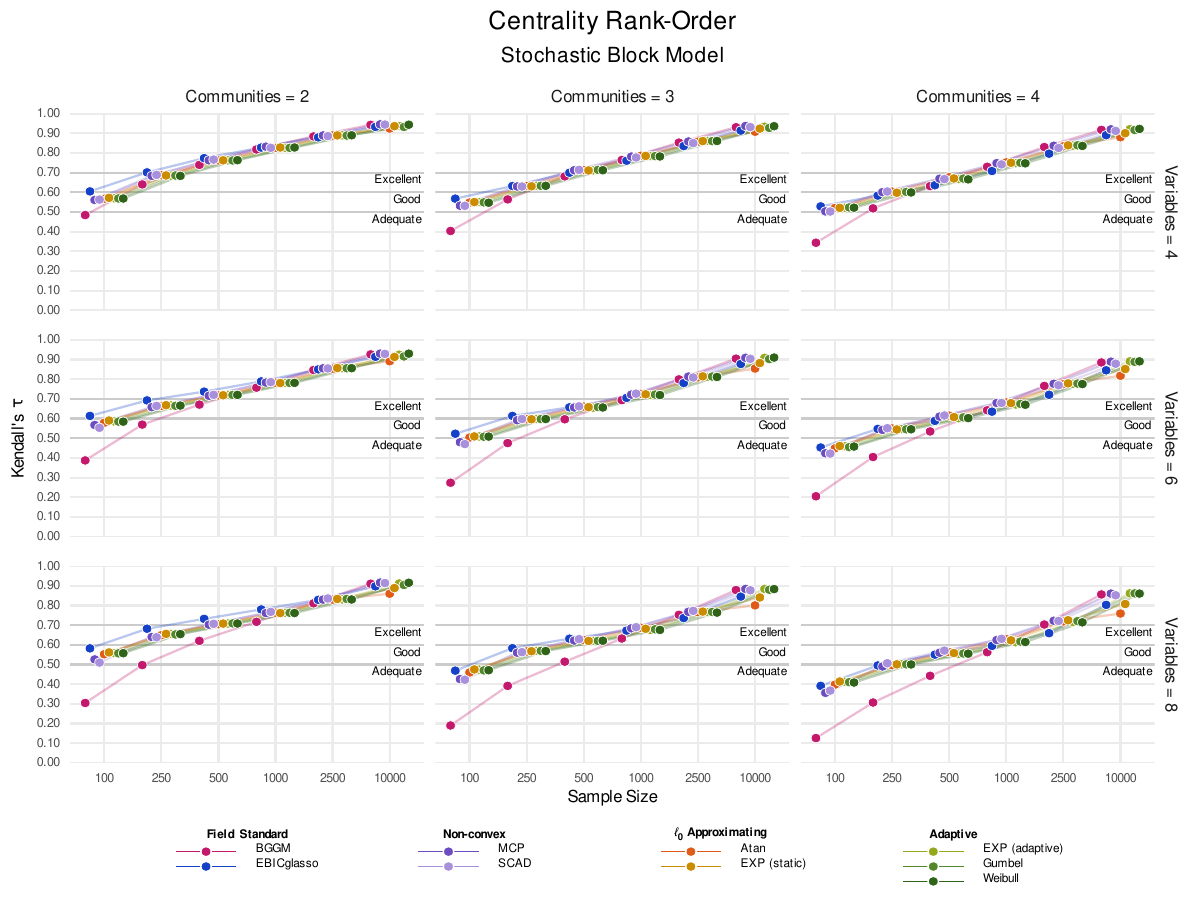}
\caption{\label{fig:sbm_kendall} Kendall's \(\tau_b\) between estimated and population node strength across methods, sample sizes, and network conditions under the stochastic block model (SBM) generating mechanism and broken down by number of communities (columns) and number of variables per community (rows). Horizontal reference lines denote qualitative benchmarks for congruence: adequate \(\tau_b \ge 0.50\), good \(\tau_b \ge 0.60\), excellent \(\tau_b \ge 0.70\). Each point represents the average \(\tau_b\) at a given sample size.}
\end{figure}

Rank-order centrality congruence followed a pattern broadly consistent with centrality bias: \(\tau_b\) improved as sample size increased and total number of variables decreased (Figure \ref{fig:sbm_kendall}), with most methods reaching at least good congruence by N = 1,000 in networks with fewer communities and variables per community but requiring N \(\ge\) 2,500 to reliably reach good congruence in larger networks. The adaptive penalties were the most consistent methods to achieve excellent congruence across conditions at larger sample sizes.

EBICglasso and BGGM diverged from the remaining methods: EBICglasso showed slightly higher \(\tau_b\) at smaller sample sizes with fewer communities, whereas BGGM showed substantially lower \(\tau_b\) across these same smaller sample sizes, degrading further as the number of communities and variables per community increased. MCP and SCAD started a step below the remaining methods at the smallest sample size (\(N\) = 100) but converged with, and ultimately surpassed, Atan and EXP (static) by the largest sample size (\(N\) = 10,000). The adaptive penalties were consistently among the top performers across conditions.

\section{Discussion}\label{discussion}

This study investigated whether EVT regularization penalties that adapt to the distribution of empirical partial correlations offer better recovery of edges, lower parameter bias, and more accurate variable selection in GGMs. Our simulation evaluated several regularization penalties across two network topologies to establish whether their performance was due to non-convexity, \(\ell_0\) approximation, or adaptivity. Our results indicated that their performance was primarily driven by \(\ell_0\) approximation with secondary improvements driven by adaptivity as sample size grew. Overall, the adaptive penalties demonstrated the most consistent performance across conditions with high specificity, sensitivity that increased with sample size, minimal parameter bias, and high rank-order centrality congruence. These results are consistent with the oracle properties we formally established.

\subsection{Simulation}\label{simulation}

EBICglasso network estimation has been the field standard since psychometric network models were first introduced to the psychological literature \citep{qgraph}. Our simulation design allowed for precise inferences into where performance differences emerged in the gradient from the EBICglasso's \(\ell_1\) penalty to non-convex and \(\ell_0\) approximating penalties (including the adaptive penalties). Consistent with previous literature \citep{isvoranu2021estimation,williams2020beyond,wysocki2019penalty}, EBICglasso's specificity declined as sample size increased, regardless of topology, indicating that false positives were not mitigated by more stable signal. Because the \(\ell_1\) penalty applies uniform shrinkage, true edges were consistently attenuated even as EBICglasso accumulated false positives, producing underestimated edge weights and substantial downward bias in node strength.

Non-convex penalties had consistently higher specificity than EBICglasso, and lower sensitivity until the largest sample size tested (\(N\) = 10,000), paralleling Williams's (\citeyear{williams2020beyond}) findings. The edge recovery patterns for SCAD and MCP revealed how the penalty derivative shape affects sensitivity and specificity (see Appendix \ref{appendix-penalties}): retaining some shrinkage on small-to-moderate edge parameters, relative to \(\ell_0\) approximators, reduced bias in false positive edge and centrality parameters.

The \(\ell_0\) approximators, Atan and EXP (static), had the highest specificity and lowest sensitivity across conditions, constraining the false positive rate to approximately 0.05 or lower regardless of sample size and topology. Sensitivity increased more slowly than for other methods and stalled once specificity reached ceiling at \(N\) = 10,000, suggesting that fixed hyperparameters exclude some signal at larger sample sizes. A consequence of this result was a larger-than-expected negative centrality bias at the largest sample sizes and total number of variables in SBMs (Figure \ref{fig:sbm_centrality_mbe}).

Among the adaptive penalties, differences were negligible across conditions: all three maintained high specificity and low parameter bias comparable to the static \(\ell_0\) approximators while continuing to accumulate sensitivity at larger sample sizes. The most theoretically controlled comparison was between static and adaptive EXP, which share identical penalty shapes and differ only in whether \(\gamma\) is fixed or calibrated to the empirical partial correlation distribution. The two variants tracked closely at smaller sample sizes, producing nearly indistinguishable edge recovery and parameter bias, but diverged as sample size grew: the static form plateaued in sensitivity while the adaptive form continued to improve, with a parallel advantage emerging in parameter bias, especially with a greater total number of variables. These results indicate that the adaptive penalty's gain over its static counterpart is concentrated at larger sample sizes.

Across adaptive penalties, there was no clear separation in performance, expected given their equivalences (Appendix \ref{appendix-evt}). On theoretical grounds, the Weibull penalty offers parameter flexibility and interpretive advantages that motivate its use in practice: its shape parameter can adjust to distributions that deviate more substantially from an exponential form. Although such deviations were rare among the empirical datasets evaluated here, this flexibility may allow the Weibull penalty to generalize more broadly (e.g., to data outside of psychology). Further, its parameters offer inference into the SNR (shape) and mean of the absolute partial correlations (scale), providing diagnostic information that can enhance the empirical understanding of partial correlations in psychology and adjacent fields.

In contrast to regularization penalties, BGGM selects edges based on posterior evidence rather than penalized likelihood, retaining only edges with sufficient evidentiary support \citep{huth2025statistical,easybgm,williams2021bayesian}. This edge selection produced more conservative network estimation relative to the regularization methods, beginning with moderate-to-high specificity and low sensitivity, which steadily increased with sample size. The rate of improvement differed by topology with greater overestimation of centrality at smaller SWN sample sizes, particularly in denser networks with more variables, before shifting to underestimation as sample size increased. This bias carried into centrality congruence, where it had the worst rank-order congruence at small sample sizes (up to \(N\) = 1,000) across both topologies and remained the worst performer under SWN topology regardless of sample size. Despite these limitations, BGGM's edge recovery trajectory of jointly increasing sensitivity and specificity is consistent with model selection consistency.

\subsection{Empirical Example}\label{empirical-example}

Our empirical example examined model selection consistency by evaluating how well each method converged to its own large sample network across increasing subsample sizes (see Appendix \ref{appendix-empirical}). These patterns largely paralleled the simulation findings: EBICglasso maintained relatively high sensitivity but exhibited steadily decreasing specificity whereas BGGM began conservatively in both metrics before ultimately achieving the highest specificity at the largest subsample size. All penalties except EBICglasso demonstrated trajectories consistent with model selection consistency. The adaptive penalties and EXP (static) converged most reliably to their large sample networks by maintaining stable specificity and steadily increasing sensitivity, suggesting that these penalties incorporate true edges conservatively without accumulating false positives as sample size grows.

Subsample comparison at a typical sample size (\(N\) = 500) revealed that Weibull occupied a consensus position between EBICglasso and BGGM, with all of its edges represented in one or both networks. The centrality rankings corroborated this intermediate position, though agreement was stronger with EBICglasso than BGGM (see Appendix \ref{appendix-empirical}). All three methods converged on the same most central node in the full sample (\(N\) = 212,265), providing a benchmark against which the \(N\) = 500 discrepancies become interpretable: at a typical sample size, method choice alone determined which node was identified as most central. This finding reinforces the broader concern that centrality, particularly the identification of a single most central node, should be interpreted with caution at sample sizes typical in psychology \citep{bringmann2019what,christensen2026network}.

\subsection{Implications for Applied Researchers}\label{implications-for-applied-researchers}

The simulation and empirical results raise concerns about contemporary field standards. The popularity of EBICglasso has persisted despite well-known limitations \citep{fan2009network,isvoranu2021estimation,williams2020beyond,wysocki2019penalty}. Our study, to our knowledge, provides the first evidence of their practical consequences: across network topologies, EBICglasso consistently underestimated edge and centrality parameters at sample sizes typical in psychology despite accumulating false positive edges. Rank-order centrality congruence was less affected, suggesting that EBICglasso's false positives may be small enough to preserve relative ordering even as they distort parameter estimates. BGGM, which is gaining popularity, exhibited the opposite pattern: conservative edge recovery with overestimation of centrality at small samples that crosses to negative bias as sample size grows, with recovery patterns that differed markedly by topology and rank-order congruence that lagged all other methods until \(N\) = 10,000. Together, these results suggest that the dominance of EBICglasso and BGGM as field standards warrants reassessment.

Such a reassessment should begin by clarifying the inferential goals of psychometric network models. Edges serve as the foundation of measurement \citep{van2024measurement} with centrality providing inferences about variable importance (cf. \citealp{bringmann2019what}). Researchers most often interpret individual edges as psychologically meaningful \citep{chambon2026network} and assess construct validity based on how nodes are connected \citep{van2024measurement}. Although a balance of high sensitivity and specificity is ideal, we argue that methods high in specificity must be preferred at sample sizes common in psychology (\(N \le\) 1,000) to ensure that estimated edges exist in the underlying network structure, accruing only as evidence accumulates \citep{huth2025statistical}. Because edges and centrality each serve distinct inferential roles in applied network research, evaluating estimation methods on edge recovery alone provides an incomplete account of practical performance.

Based on these criteria, MCP, \(\ell_0\) approximating, and adaptive penalties are the most defensible choice, with BGGM only in larger sample sizes (\(N \ge\) 1,000). The adaptive penalties, in particular, were consistently among the top performers across every metric evaluated, demonstrating model selection consistency with high specificity even at small sample sizes (\(N\) = 100) and a considerable number of variables (e.g., \(p\) = 30; \citealp{chambon2026network,huth2025statistical}). Estimated edges are therefore unlikely to be false positives, with those recovered at larger sample sizes reflecting accumulated evidence rather than noise. Applied researchers can therefore be confident that their estimated networks, and the centrality parameters derived from them, are congruent with the underlying network structure.

\subsection{Limitations}\label{limitations}

This study was restricted to continuous data to evaluate GGM estimation methods under their standard assumptions (e.g., multivariate normal distribution). Importantly, the majority of psychological datasets are categorical, limiting the generalizability of our results \citep{huth2025statistical}. Polychoric and tetrachoric correlations offer feasible extensions to ordinal data, and future simulations should evaluate whether the present results hold with categorical data. Future research should also examine whether the penalties employed here can be adapted to Ising model estimation \citep{xue2012nonconcave}, given that current applications rely on a variant of the $\ell_1$ penalty \citep{van2014new}. Finally, although we provide formal proof of the EVT penalties' oracle properties, they fail to satisfy Fan et al.'s (\citeyear{fan2014strong}) condition (iv) for strong oracle properties, which specifies that the penalty's derivative must be exactly zero above a finite threshold, whereas the EVT penalties decay exponentially and approach zero asymptotically.

\subsection{Conclusion}\label{conclusion}

Accurate network estimation is prerequisite to the measurement validity of network models in psychology, where construct validation depends on whether nodes are connected and variable importance is inferred from these connections \citep{chambon2026network,van2024measurement}. Edge selection is therefore fundamentally a variable selection problem where the pairwise relationships determine which variables are deemed important. We establish a set of data-adaptive regularization penalties grounded in EVT that calibrate to the empirical partial correlation distribution, with results demonstrating that this adaptivity allows for high specificity while accumulating true positive edges as sample size increases. This property is consequential for both construct interpretation and model parsimony, as sparser and more accurate networks yield cleaner construct representations and more defensible inferences about variable importance. The Weibull penalty in particular offers an interpretable bridge between EVT and the empirical regularities of psychological partial correlations, making data-adaptive regularization an accessible and principled default for applied researchers seeking network estimates that are parsimonious and interpretable.

\newpage

\bibliographystyle{apacite}
\bibliography{Christensen_Citations}

\newpage

\appendix

\section{Distribution of Empirical Partial Correlations}\label{appendix-distribution}

\hspace{\parindent}

Using the 290 empirical datasets from Huth and colleagues (3 datasets were duplicated), we fit several distributions, using maximum likelihood estimators (MLE), to the absolute value of each dataset's partial correlations. These distributions were selected to span a range of motivations: Beta for its theoretical precedent as the null distribution of squared partial correlations under multivariate normality \citep{muirhead1982aspects}, Gamma for its flexibility, Normal as a baseline, and Exponential, Log-Normal, and Weibull for their visual congruence with the empirical densities (left in Figure \ref{fig:distributions}). Log-likelihood (\(\cal{L}\)) was computed for the fit of each distribution to each dataset's absolute partial correlations. The Beta distribution had the largest log-likelihood for the highest proportion of datasets (32.4\%) followed by Weibull (27.9\%) and Gamma (26.9\%).

\setcounter{figure}{0}
\renewcommand{\thefigure}{A\arabic{figure}}

\begin{figure}[H]
\centering
\includegraphics[width=\textwidth]{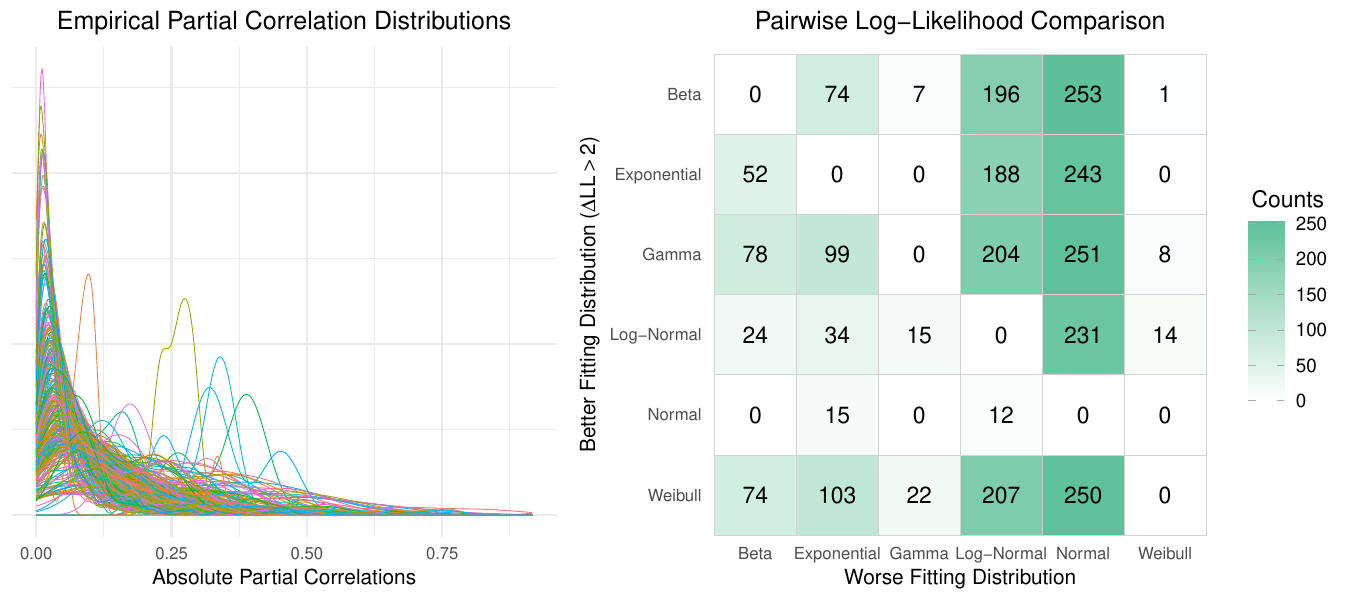}
\caption{\label{fig:distributions} Empirical distributions of the absolute partial correlations where each density curve represents a single sample in Huth et al.'s (2025) dataset collection (left). Comparision of log-likelihood between estimated maximum likelihood estimator distributions with the values going across representing the number of samples where the y-axis distribution fit at least 2 log-likelihood units greater than the x-axis distribution (right).}
\end{figure}

Looking deeper into the pairwise comparisons between distributions, a more nuanced pattern emerges (right in Figure \ref{fig:distributions}). When comparing cases where \(\Delta\cal{L}\) was greater than 2 between distributions, the Weibull distribution becomes favored with the Gamma distribution closely trailing.\footnote{The \(\Delta\mathcal{L} > 2\) threshold was selected for its interpretation of distribution complexity: a difference of two log-likelihood units means the worse fitting distribution would require at least one additional free parameter to achieve equivalent fit, following from AIC (\(-2\mathcal{L} + 2p\)). For the five two-parameter distributions compared here, \(\Delta\mathcal{L} > 2\) corresponds to \(\Delta\text{AIC} > 4\), which is conservative relative to the conventional \(\Delta\text{AIC} > 2\) threshold for substantial evidence \citep{burnham2002model}. The Exponential distribution, having one fewer free parameter than all other candidate distributions, is the sole exception: \(\Delta\mathcal{L} > 2\) corresponds exactly to \(\Delta\text{AIC} > 2\) for those comparisons.} Despite the Beta distribution having larger log-likelihoods in a higher proportion of datasets, its advantage over the Weibull distribution was relatively minimal (only one dataset fitting with \(\cal{L}\) \textgreater{} 2). Instead, the Weibull distribution often fit as well as or better than all other distributions with only a small subset fitting other distributions better (Beta = 1, Gamma = 8, Log-Normal = 14). This result extends Christensen and Choi's (\citeyear{christensen2026network}) finding beyond regularized edge weights to the raw partial correlation distributions, demonstrating that the Weibull distribution holds prior to any model selection step and generalizes across a vast number of psychological datasets that vary in content, number of variables (\(Mdn = 14, range = 3-86\)), and sample size (\(Mdn = 461, range = 23-388,286\)).

\newpage

\section{Extreme Value Distributions}\label{appendix-evt}

\begin{flushleft}
\textbf{Table B1} 

\textit{Forms and Connections between Extreme Value Distributions}
\end{flushleft}

\begin{table}
\centering
\renewcommand{\arraystretch}{1.5}
\begin{threeparttable}
\begin{tabular}{ccc}

\textbf{Distribution} & \textbf{PDF} & \textbf{CDF} \\ \hline

Exponential & \( \frac{1}{\gamma} e^{-\frac{|x|}{\gamma}} \) & \( 1 - e^{-\frac{|x|}{\gamma}} \) \\

Weibull & \( \left(\frac{k}{\gamma}\right)\left(\frac{|x|}{\gamma}\right)^{k-1} e^{-\left(\frac{|x|}{\gamma}\right)^k} \) & \( 1 - e^{-\left(\frac{|x|}{\gamma}\right)^k} \) \\

Gumbel & \( \frac{e^{\frac{|x|}{\gamma} - e^{\frac{|x|}{\gamma}}}}{\gamma} \) & \( 1 - e^{-e^{\frac{|x|}{\gamma}}} \) \\

Weibull \( \rightarrow \) Exponential & \( k = 1 \) & \( k = 1 \) \\

Weibull \( \rightarrow \) Gumbel & \( W_{\text{PDF}}\!\left(e^{x},\, k = \tfrac{1}{\gamma},\, \gamma = 1\right) \cdot e^{x} \) & \( W_{\text{CDF}}\!\left(e^{x},\, k = \tfrac{1}{\gamma},\, \gamma = 1\right) \) \\

\bottomrule
\end{tabular}
\begin{tablenotes}
\footnotesize
\vspace{0.25em}
\item \textit{Note}. All distributions are parameterized without their location parameter, which is set to zero.
\end{tablenotes}
\end{threeparttable}
\end{table}

\newpage

\section{Oracle Property Proofs}\label{appendix-proof}

\subsection{Proof of Theorem~\ref{thm:high-dimensional-rate}}

\begin{proof}
Without loss of generality, we assume that
\[
(\boldsymbol\Sigma_0)_{ii}=1,\qquad i=1,\ldots,p.
\]

Define
\[
\widehat{\mathbf D}
=\operatorname{diag}(S_{11}^{1/2},\ldots,S_{pp}^{1/2}),\qquad
\widehat{\mathbf R}=\widehat{\mathbf D}^{-1}\mathbf S
\widehat{\mathbf D}^{-1}.
\]
This is the preceding penalized likelihood applied after marginal
standardization. The penalized log-likelihood criterion used below is
\begin{equation}
Q_n(\mathbf K)=\log\det(\mathbf K)
-\operatorname{tr}(\widehat{\mathbf R}\mathbf K)
-\sum_{i<j}p_{\lambda_{ij}}(|k_{ij}|),
\qquad \mathbf K\succ0.
\label{eq:standardized-objective}
\end{equation}

and write
\[
L_n(\mathbf K)=-Q_n(\mathbf K)
=-\log\det(\mathbf K)+\operatorname{tr}(\widehat{\mathbf R}\mathbf K)
+\sum_{i<j}p_{\lambda_{ij}}(|k_{ij}|).
\]
\begin{equation}
Q_n(\mathbf K)
=\ell_n(\mathbf K)-\sum_{i<j}p_{\lambda_{ij}}(|k_{ij}|),
\qquad
\ell_n(\mathbf K)
=\log\det(\mathbf K)
-\operatorname{tr}(\widehat{\mathbf R}\mathbf K).
\label{eq:proof-objective}
\end{equation}
Let $\mathbf E=\widehat{\mathbf R}-\boldsymbol\Sigma_0$.The bounded-spectrum part of condition (i) keeps the variances of
$X_iX_j$ uniformly bounded. Standard Gaussian concentration therefore gives,
for sufficiently small $t>0$,
\[
P\!\left(|S_{ij}-(\boldsymbol\Sigma_0)_{ij}|>t\right)
\leq c_1\exp(-c_2nt^2),
\]
where $c_1$ and $c_2$ do not depend on $i$ or $j$. Taking a union bound over
the $p^2$ entries and setting $t$ proportional to
$\sqrt{\log p/n}$ yields
\[
\|\mathbf S-\boldsymbol\Sigma_0\|_{\max}=O_p(\tau_n).
\] Since
$\widehat R_{ij}=S_{ij}/\sqrt{S_{ii}S_{jj}}$ and
$\max_i|S_{ii}-1|=O_p(\tau_n)=o_p(1)$, the same bound holds for the sample
correlations:
\begin{equation}
\|\mathbf E\|_{\max}=O_p(\tau_n).
\label{eq:proof-rate-concentration}
\end{equation}
In addition, $E_{ii}=0$ for every $i$.
The assumption $s\log p=o(n)$ gives $r_n=o(1)$; when $p\to\infty$, the
bound also gives
$r_n^2=o(1)$.

 Let
$\boldsymbol\Delta=\boldsymbol\Delta^{\mathsf T}$ and consider
$\mathbf K=\mathbf K_0+\boldsymbol\Delta$ on the sphere
$\|\boldsymbol\Delta\|_F=Mr_n$, where $M>0$ is fixed. Because
$\|\boldsymbol\Delta\|_2\leq\|\boldsymbol\Delta\|_F=o(1)$, the ball
shrinks to zero and every matrix in the ball is eventually positive definite. Expanding
the log determinant with its integral remainder gives
\begin{align}
&L_n(\mathbf K_0+\boldsymbol\Delta)
-L_n(\mathbf K_0)\notag\\
&\quad=\operatorname{tr}(\mathbf E\boldsymbol\Delta)
+R_n(\boldsymbol\Delta)
+\sum_{i<j}\{p_{\lambda_{ij}}(|k_{ij}^0+\delta_{ij}|)
-p_{\lambda_{ij}}(|k_{ij}^0|)\},
\label{eq:proof-rate-expansion}
\end{align}
where the second-order remainder of $-\log\det(\mathbf K)$ is
\[
R_n(\boldsymbol\Delta)
=\operatorname{vec}(\boldsymbol\Delta)^{\mathsf T}
\left\{\int_0^1(1-u)
(\mathbf K_0+u\boldsymbol\Delta)^{-1}
\otimes(\mathbf K_0+u\boldsymbol\Delta)^{-1}\,du\right\}
\operatorname{vec}(\boldsymbol\Delta).
\]
The upper eigenvalue bound on $\mathbf K_0$ and $r_n=o(1)$ imply that, for
some $c_0>0$,
\begin{equation}
R_n(\boldsymbol\Delta)
\geq c_0\|\boldsymbol\Delta\|_F^2
=c_0M^2r_n^2
\label{eq:proof-rate-curvature}
\end{equation}
uniformly on the sphere.

Let $\boldsymbol\Delta_{\mathcal T}$ denote the part of
$\boldsymbol\Delta$ corresponding to the true edges. Since $E_{ii}=0$, the
diagonal entries make no contribution. The contribution from the $s$ true
edges is therefore bounded by
\begin{align}
\left|\operatorname{tr}
(\mathbf E\boldsymbol\Delta_{\mathcal T})\right|
&\leq C_1\|\mathbf E\|_{\max}\sqrt{s}\,
\|\boldsymbol\Delta\|_F \notag\\
&=O_p(Mr_n^2),
\label{eq:proof-rate-active-score}
\end{align}
because $\|\mathbf E\|_{\max}=O_p(\tau_n)$,
$\|\boldsymbol\Delta\|_F=Mr_n$, and $\sqrt{s}\tau_n=r_n$.

For an inactive pair, $k_{ij}^0=0$. Its contribution from the first and last
terms in \eqref{eq:proof-rate-expansion} is bounded below by
\begin{equation}
\sum_{(i,j)\in\mathcal A^c}
\left\{p_{\lambda_{ij}}(|\delta_{ij}|)
-2\|\mathbf E\|_{\max}|\delta_{ij}|\right\}.
\label{eq:proof-rate-null-combination}
\end{equation}
Every $|\delta_{ij}|$ is at most $Mr_n$. By concavity, the penalty
derivative is nonincreasing, and hence
\[
p_{\lambda_{ij}}(|\delta_{ij}|)
=\int_0^{|\delta_{ij}|}p_{\lambda_{ij}}^{\prime}(t)\,dt
\geq |\delta_{ij}|\inf_{0<t\leq Mr_n}
p_{\lambda_{ij}}^{\prime}(t).
\]
Condition \eqref{eq:null-derivative-common} and
\eqref{eq:proof-rate-concentration} therefore show that
\eqref{eq:proof-rate-null-combination} is nonnegative with probability
tending to one.

The assumption $\frac{k_{min}}{r_n}\to\infty$ says that the smallest true edge is much larger $Mr_n$. Thus the penalty
derivative remains below $\eta_n$. Therefore,
\begin{align}
&\sum_{(i,j)\in\mathcal A}
\{p_{\lambda_{ij}}(|k_{ij}^0+\delta_{ij}|)
-p_{\lambda_{ij}}(|k_{ij}^0|)\}\notag\\
&\qquad\geq-\eta_n\sum_{(i,j)\in\mathcal A}|\delta_{ij}|
\geq-\eta_n\sqrt{s}\,\|\boldsymbol\Delta\|_F
=-o(Mr_n^2).
\label{eq:proof-rate-active-penalty}
\end{align}
Indeed,
$\sqrt{s}\eta_n/r_n=(\sqrt n\eta_n)/\sqrt{\log p}=o(1)$ by (iii).

Combining \eqref{eq:proof-rate-curvature},
\eqref{eq:proof-rate-active-score},
\eqref{eq:proof-rate-null-combination}, and
\eqref{eq:proof-rate-active-penalty} yields
\[
\inf_{\|\boldsymbol\Delta\|_F=Mr_n}
\{L_n(\mathbf K_0+\boldsymbol\Delta)
-L_n(\mathbf K_0)\}
\geq c_0M^2r_n^2-O_p(Mr_n^2).
\]
 For any small $\epsilon>0$, choose $M$ sufficiently large so that, With probability at least $1-\epsilon$,
the right-hand side is positive. The continuous function $L_n$ attains a
minimum over the closed ball, and no minimizer can lie on its boundary. An interior minimizer is a local
minimizer on the positive-definite cone. Equivalently, $Q_n$ has a local
maximizer in the ball.

Finally, condition (i) ensures that converting the precision matrix into
partial correlations does not change the order of the estimation error.
In particular, for some constant $C>0$, with probability tending to one,
\[
\|\widehat{\boldsymbol\Omega}-\boldsymbol\Omega_0\|_F
\leq C\|\widehat{\mathbf K}-\mathbf K_0\|_F.
\]
The result now follows from \eqref{eq:high-dimensional-rate}.

Therefore, the proof of Theorem~\ref{thm:high-dimensional-rate} is complete.
\end{proof}

\subsection{Proof of Theorem~\ref{thm:oracle-fixed-p}}

\begin{proof}

Let
$\widetilde{\mathbf K}^{\,o}$ be the constrained maximum likelihood estimator
under the restriction that $k_{ij}=0$ for all $(i,j)\notin\mathcal A$, with associated partial correlation matrix having entries
$\widetilde\omega_{ij}^{\,o}$.

We shall construct a local maximizer with the correct support and then show
that it is first-order equivalent to the oracle estimator:
\begin{equation}
\sqrt n\left\{
(\widehat\omega_{ij})_{(i,j)\in\mathcal A}
-(\widetilde\omega_{ij}^{\,o})_{(i,j)\in\mathcal A}
\right\}
=o_p(1).
\label{eq:oracle-equivalence}
\end{equation}
We first restrict the parameter space to the true model,
\[
\mathcal M_{\mathcal A}
=\{\mathbf K\succ0:k_{ij}=0\text{ for }(i,j)\notin\mathcal A,\ i<j\}.
\]
Although $\ell_n$ is written in terms of the sample correlation matrix, it
is equivalent to the usual Gaussian likelihood. In fact, setting
$\mathbf L=\widehat{\mathbf D}^{-1}\mathbf K
\widehat{\mathbf D}^{-1}$ gives
\[
\log\det(\mathbf L)-\operatorname{tr}(\mathbf S\mathbf L)
=\ell_n(\mathbf K)-2\log\det(\widehat{\mathbf D}).
\]
This one-to-one transformation preserves both the zero pattern and the
partial correlations. Hence maximizing $\ell_n$ over
$\mathcal M_{\mathcal A}$ is equivalent to fitting the Gaussian model with
the true zero restrictions.

Let $\boldsymbol\vartheta$ collect the free entries of $\mathbf K$ in this
model, including the diagonal entries, and let
$\widetilde{\boldsymbol\vartheta}^{\,o}$ denote the oracle maximum likelihood
estimator. Since $p$ is fixed and the oracle information matrix is positive
definite, standard likelihood theory gives
\begin{equation}
\|\widetilde{\boldsymbol\vartheta}^{\,o}
-\boldsymbol\vartheta_0\|_2=O_p(n^{-1/2}).
\label{eq:proof-oracle-mle-rate}
\end{equation}

Fix $R>0$ and consider the closed ball
\[
\mathcal B_n(R)=\left\{\boldsymbol\vartheta:
\|\boldsymbol\vartheta-
\widetilde{\boldsymbol\vartheta}^{\,o}\|_2\leq R/\sqrt n\right\}.
\]
The oracle estimator is consistent and the radius of
$\mathcal B_n(R)$ is $R/\sqrt n$. Since $\sqrt{n}k_{min}\to\infty$, all active entries in this ball, and along the segments
from its center, remain larger than $k_{\min}/2$ with probability tending to one.
Moreover, the likelihood score vanishes at
$\widetilde{\boldsymbol\vartheta}^{\,o}$. The linear term in the Taylor
expansion therefore disappears, while positive definiteness of the oracle
information makes the quadratic term negative. Hence, for some $c_2>0$,
\begin{equation}
\sup_{\boldsymbol\vartheta\in\partial\mathcal B_n(R)}
\{\ell_n(\boldsymbol\vartheta)
-\ell_n(\widetilde{\boldsymbol\vartheta}^{\,o})\}
\leq-\frac{c_2R^2}{n}+o_p(n^{-1}).
\label{eq:proof-oracle-boundary}
\end{equation}
The mean-value theorem and Cauchy--Schwarz give, uniformly on the boundary,
\begin{align}
&\left|\sum_{(i,j)\in\mathcal A}
\{p_{\lambda_{ij}}(|k_{ij}|)
-p_{\lambda_{ij}}(|\widetilde k_{ij}^{\,o}|)\}\right|\notag\\
&\qquad\leq \eta_n\sum_{(i,j)\in\mathcal A}
|k_{ij}-\widetilde k_{ij}^{\,o}|
\leq\frac{R\sqrt{s}}{\sqrt n}\eta_n=o_p(n^{-1})
\label{eq:proof-oracle-active-penalty}
\end{align}
by condition (iii). Combining
\eqref{eq:proof-oracle-boundary} and
\eqref{eq:proof-oracle-active-penalty} shows that the restricted penalized
objective is smaller on the boundary than at the oracle estimator, with
probability tending to one. Its maximum over the closed ball is therefore
attained in the interior. Denote one such restricted local maximizer by
$\widehat{\mathbf K}^{\,r}$. The radius of the ball and
\eqref{eq:proof-oracle-mle-rate} give
\begin{equation}
\|\widehat{\mathbf K}^{\,r}-\mathbf K_0\|_F
=O_p(n^{-1/2}).
\label{eq:proof-restricted-rate}
\end{equation}

It remains to verify that this oracle-space local maximizer is also a local
maximizer of the full objective. Since $p$ is fixed,
\[
\|\widehat{\mathbf R}-\boldsymbol\Sigma_0\|_{\max}=O_p(n^{-1/2}).
\]
The inverse identity
\[
(\widehat{\mathbf K}^{\,r})^{-1}-\boldsymbol\Sigma_0
=(\widehat{\mathbf K}^{\,r})^{-1}
(\mathbf K_0-\widehat{\mathbf K}^{\,r})\boldsymbol\Sigma_0
\]
and \eqref{eq:proof-restricted-rate} imply
\begin{equation}
\|\nabla\ell_n(\widehat{\mathbf K}^{\,r})\|_{\max}
=O_p(n^{-1/2}).
\label{eq:proof-oracle-score}
\end{equation}
For fixed $D>0$, set
\[
b_n(D)=\min_{(i,j)\in\mathcal A^c}
\inf_{0<t\leq D/\sqrt n}p_{\lambda_{ij}}^{\prime}(t).
\]
Because $p$ is fixed, $r_n=c_r/\sqrt n$ and
$\tau_n=c_\tau/\sqrt n$ for positive constants $c_r,c_\tau$. Applying
\eqref{eq:null-derivative-common} with $M=D/c_r$ gives
$\sqrt n\,b_n(D)\to\infty$. Together with
\eqref{eq:proof-oracle-score}, this yields
\begin{equation}
\frac{\|\nabla\ell_n(\widehat{\mathbf K}^{\,r})\|_{\max}}
{b_n(D)}\longrightarrow0
\quad\text{in probability}.
\label{eq:proof-oracle-score-gap}
\end{equation}

On an event whose probability tends to one, continuity of matrix inversion
on the positive-definite cone provides a sufficiently small neighborhood of
$\widehat{\mathbf K}^{\,r}$ in which every matrix is positive definite,
\begin{equation}
2\max_{i<j}|\{\mathbf K^{-1}-\widehat{\mathbf R}\}_{ij}|
<\frac12b_n(D),
\label{eq:proof-oracle-uniform-score}
\end{equation}
and every inactive perturbation has magnitude at most $D/\sqrt n$. Shrink
the neighborhood if necessary so that its projection onto the diagonal and
active coordinates lies within the restricted local-maximum neighborhood.

Now consider a small symmetric perturbation
$\boldsymbol\Delta=\boldsymbol U+\boldsymbol V$, where
$\boldsymbol U$ changes only the diagonal and active entries, while
$\boldsymbol V$ changes only the inactive entries. Take the neighborhood
small enough that the entire segment
$\widehat{\mathbf K}^{\,r}+\boldsymbol U+t\boldsymbol V$,
$0\leq t\leq1$, remains inside it. Since both
$\widehat{\mathbf K}^{\,r}$ and $\boldsymbol U$ are zero on the inactive
coordinates, adding $\boldsymbol V$ changes each such coordinate from zero
to $v_{ij}$. Because the penalty derivative is nonincreasing,
\begin{equation}
p_{\lambda_{ij}}(|v_{ij}|)
=\int_0^{|v_{ij}|}p_{\lambda_{ij}}^{\prime}(t)\,dt
\geq b_n(D)|v_{ij}|.
\label{eq:proof-oracle-penalty-increment}
\end{equation}
The derivative of the likelihood with respect to one free symmetric
off-diagonal entry is
$2\{\mathbf K^{-1}-\widehat{\mathbf R}\}_{ij}$. Integrating
along the segment and applying
\eqref{eq:proof-oracle-uniform-score}--
\eqref{eq:proof-oracle-penalty-increment} gives
\begin{align}
&Q_n(\widehat{\mathbf K}^{\,r}+\boldsymbol U+\boldsymbol V)
-Q_n(\widehat{\mathbf K}^{\,r}+\boldsymbol U)\notag\\
&\qquad\leq-\frac12b_n(D)
\sum_{(i,j)\in\mathcal A^c}|v_{ij}|.
\label{eq:proof-oracle-inactive-decrease}
\end{align}
This inequality is strict whenever $\boldsymbol V\neq\mathbf0$. The
restricted local-maximum property gives
$Q_n(\widehat{\mathbf K}^{\,r}+\boldsymbol U)
\leq Q_n(\widehat{\mathbf K}^{\,r})$. Hence
$\widehat{\mathbf K}^{\,r}$ is a local maximizer of the full
objective. Set $\widehat{\mathbf K}
=\widehat{\mathbf K}^{\,r}$.

The construction gives $\widehat k_{ij}=0$ for every inactive pair. For an
active edge,
\[
|\widehat k_{ij}|
\geq|k_{ij}^0|
-\|\widehat{\mathbf K}-\mathbf K_0\|_{\max}.
\]
Since $\sqrt{n}k_{\min}\to\infty$ and \eqref{eq:proof-restricted-rate} makes the second term
$k_{min}$, every active estimate is nonzero with probability tending to one.
This proves \eqref{eq:oracle-selection}.

Finally, we compare the penalized estimator with the oracle maximum
likelihood estimator. On the selection event, both estimators belong to the
same oracle model and satisfy
\[
\nabla\ell_n(\widehat{\boldsymbol\vartheta})
-\boldsymbol d_n(\widehat{\boldsymbol\vartheta})=\mathbf0,
\qquad
\nabla\ell_n(\widetilde{\boldsymbol\vartheta}^{\,o})=\mathbf0,
\]
where the diagonal components of $\boldsymbol d_n$ are zero and each active
component is the signed derivative of the corresponding penalty.

Subtracting the two score equations and expanding the likelihood score
between the two estimators gives
\begin{equation}
\widehat{\boldsymbol\vartheta}
-\widetilde{\boldsymbol\vartheta}^{\,o}
=-\mathbf H_n^{-1}
\boldsymbol d_n(\widehat{\boldsymbol\vartheta}),
\label{eq:proof-oracle-score-expansion}
\end{equation}
where
\[
\mathbf H_n
=-\int_0^1\nabla^2\ell_n\!\left(
\widetilde{\boldsymbol\vartheta}^{\,o}
+t\{\widehat{\boldsymbol\vartheta}
-\widetilde{\boldsymbol\vartheta}^{\,o}\}\right)\,dt.
\]
Since both estimators are consistent and the oracle information matrix is
positive definite,
\[
\|\mathbf H_n^{-1}\|_2=O_p(1).
\]
Moreover, every estimated active edge has magnitude at least $k_{\min}/2$ with
probability tending to one. Since $s$ is fixed,
\[
\|\boldsymbol d_n(\widehat{\boldsymbol\vartheta})\|_2
\leq \sqrt{s}\max_{(i,j)\in\mathcal A}
\sup_{t\geq k_{\min}/2}p_{\lambda_{ij}}^{\prime}(t).
\]
Condition (iii) and
\eqref{eq:proof-oracle-score-expansion} therefore imply
\begin{equation}
\sqrt n\,\|\widehat{\boldsymbol\vartheta}
-\widetilde{\boldsymbol\vartheta}^{\,o}\|_2=o_p(1).
\label{eq:proof-oracle-free-equivalence}
\end{equation}

The oracle estimator is the maximum likelihood estimator under the true
zero restrictions. Standard fixed-dimensional likelihood theory gives
\begin{equation}
\sqrt n\left\{
(\widetilde\omega_{ij}^{\,o})_{(i,j)\in\mathcal A}
-(\omega_{ij}^0)_{(i,j)\in\mathcal A}
\right\}
\Longrightarrow
N_s\!\left(\mathbf0,
(\mathcal I_{\mathcal A}^{\,o})^{-1}\right).
\label{eq:proof-oracle-normality}
\end{equation}
The partial-correlation map
\[
(k_{ij},k_{ii},k_{jj})
\longmapsto-\frac{k_{ij}}{\sqrt{k_{ii}k_{jj}}}
\]
is continuously differentiable near $\mathbf K_0$. Hence
\eqref{eq:proof-oracle-free-equivalence} also gives
\eqref{eq:oracle-equivalence}. Combining this result with
\eqref{eq:proof-oracle-normality} and Slutsky's theorem proves
\eqref{eq:oracle-asymptotic-normality}.

Therefore, the proof of Theorem~\ref{thm:oracle-fixed-p} is complete.
\end{proof}

\subsection{Proof of Corollary~\ref{cor:oracle-exp}}

\begin{proof}
For $t>0$,
\[
p^{\prime}(t;\lambda_n,\gamma_n)
=\frac{\lambda_n}{\gamma_n}e^{-t/\gamma_n}.
\]
The penalty is continuous, vanishes at zero, and has a positive
decreasing derivative, so (ii) holds. For every fixed
$D>0$,
\[
\sqrt n\inf_{0<t\leq D/\sqrt n}p^{\prime}(t;\lambda_n,\gamma_n)
=\sqrt n\frac{\lambda_n}{\gamma_n}
\exp\!\left(-\frac{D}{\sqrt n\,\gamma_n}\right)\longrightarrow\infty
\]
by \eqref{eq:oracle-exp-tuning}. Monotonicity also gives
\[
\sqrt n\sup_{t\geq k_{\min}/2}p^{\prime}(t;\lambda_n,\gamma_n)
=\sqrt n\frac{\lambda_n}{\gamma_n}e^{-k_{\min}/(2\gamma_n)}
\longrightarrow0
\]
by \eqref{eq:oracle-exp-tuning}. Thus (iii) holds in the fixed-$p$ regime and
the theorem applies.

Therefore, the proof of Corollary~\ref{cor:oracle-exp} is complete.
\end{proof}

\subsection{Proof of Corollary~\ref{cor:oracle-weibull}}

\begin{proof}
The modified Weibull derivative is
\begin{equation}
p^{\prime}(t;\lambda_n,\gamma_n,k)
=\frac{\lambda_n k}{\gamma_n}
\left(\frac{t}{\gamma_n}\right)^{k-1}
\exp\!\left\{-\left(\frac{t}{\gamma_n}\right)^k\right\}.
\label{eq:proof-weibull-derivative}
\end{equation}
For $0<k\leq1$, this derivative decreases on $(0,\infty)$. For $k>1$, the
ordinary derivative increases up to
$x_n^*=\gamma_n\{(k-1)/k\}^{1/k}$ and decreases thereafter; the modification
replaces the increasing portion by its value at $x_n^*$. The modified
penalty is continuous, vanishes at zero, and has a nonnegative,
 nonincreasing derivative. Thus (ii) holds.

If $0<k\leq1$, monotonicity gives
\begin{align*}
&\sqrt n\inf_{0<u\leq D/\sqrt n}
p^{\prime}(u;\lambda_n,\gamma_n,k)\\
&\quad=\sqrt n\frac{\lambda_n k}{\gamma_n}
\left(\frac{D}{\sqrt n\,\gamma_n}\right)^{k-1}
\exp\!\left[-\left(\frac{D}{\sqrt n\,\gamma_n}\right)^k\right].
\end{align*}
The exponential factor tends to one; the power factor equals one when
$k=1$ and diverges when $0<k<1$. Hence the display diverges by
\eqref{eq:oracle-weibull-tuning}. If $k>1$, then
$D/\sqrt n<x_n^*$ eventually because $\sqrt n\gamma_n\to\infty$. The
modified derivative is constant on this interval, and
\begin{align*}
&\sqrt n\inf_{0<u\leq D/\sqrt n}
p^{\prime}(u;\lambda_n,\gamma_n,k)\\
&\quad=\sqrt n\frac{\lambda_n k}{\gamma_n}
\left(\frac{k-1}{k}\right)^{(k-1)/k}
\exp\!\left(-\frac{k-1}{k}\right)\longrightarrow\infty.
\end{align*}
This proves the first part of (iii).

For $0<k\leq1$, the supremum over $t\geq k_{\min}/2$ is attained at $k_{\min}/2$.
When $k>1$, $\gamma_n\to0$ and $\frac{1}{\gamma_n}k_{min}\to\infty$ imply
$k_{\min}/2>x_n^*$ eventually, so the same is true. Substitution in
\eqref{eq:proof-weibull-derivative} gives the expression in
\eqref{eq:oracle-weibull-tuning}, which tends to zero. The active-edge part
of (iii) follows.

Therefore, the proof of Corollary~\ref{cor:oracle-weibull} is complete.
\end{proof}

\subsection{Proof of Corollary~\ref{cor:oracle-gumbel}}

\begin{proof}
For $u\geq0$, let
\[
h(u)=e^{-u-e^{-u}}+e^{u-e^u}.
\]
The folded Gumbel derivative is
\[
p^{\prime}(t;\lambda_n,\gamma_n)
=\frac{\lambda_n}{\gamma_n}h(t/\gamma_n).
\]
For $u>0$,
\[
h^{\prime}(u)=e^{-u-e^{-u}}(-1+e^{-u})
+e^{u-e^u}(1-e^u)<0,
\]
and $h(0)=2/e$. The folded Gumbel penalty is continuous, vanishes at zero,
and has a positive decreasing derivative, so (ii) holds. For every fixed
$D>0$,
\[
\sqrt n\inf_{0<t\leq D/\sqrt n}p^{\prime}(t;\lambda_n,\gamma_n)
=\sqrt n\frac{\lambda_n}{\gamma_n}
h\!\left(\frac{D}{\sqrt n\,\gamma_n}\right)\longrightarrow\infty.
\]
For $u\geq0$, $e^{-u-e^{-u}}\leq e^{-u}$ and $e^u\geq2u$ gives
$e^{u-e^u}\leq e^{-u}$. Hence $h(u)\leq2e^{-u}$ and
\[
\sqrt n\sup_{t\geq k_{\min}/2}p^{\prime}(t;\lambda_n,\gamma_n)
\leq2\sqrt n\frac{\lambda_n}{\gamma_n}e^{-k_{\min}/(2\gamma_n)}
\longrightarrow0
\]
by \eqref{eq:oracle-gumbel-tuning}.

Therefore, the proof of Corollary~\ref{cor:oracle-gumbel} is complete.
\end{proof}

\newpage

\section{Penalty and Derivative Shapes}\label{appendix-penalties}

\subsection{Penalties}\label{penalties}

\setcounter{figure}{0}
\renewcommand{\thefigure}{D\arabic{figure}}

\begin{figure}[H]
\centering
\includegraphics[width=4.75in]{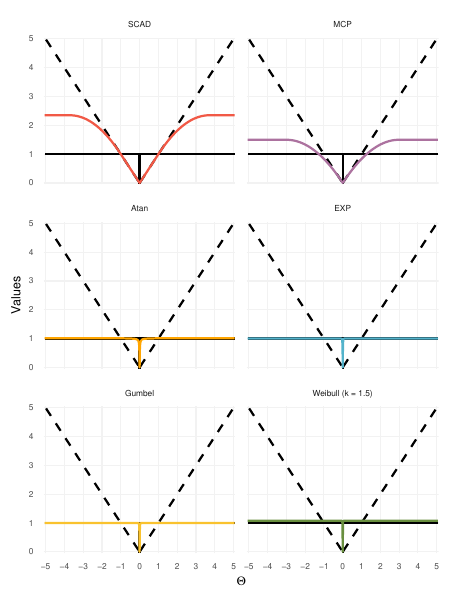}
\caption{\label{fig:penalties} Depictions of the penalty derivatives in this study organized by non-convex (top row), \(ell_0\) approximators (middle row), and adaptive penalties (bottom row). Hyperparameters were set to \(\lambda = 1\) and defaults for \(\gamma\). Weibull's shape parameter (\(k\)) was set to 1.5 to demonstrate how it influences the derivative relative to EXP. The dashed black line across the top of each panel represnts the \(\ell_1\) penalty and the solid black line down the y-axis and across the x-axis of each panel represents the \(\ell_0\) penalty.}
\end{figure}

\subsection{Derivatives}\label{derivatives}

\begin{figure}[H]
\centering
\includegraphics[width=4.75in]{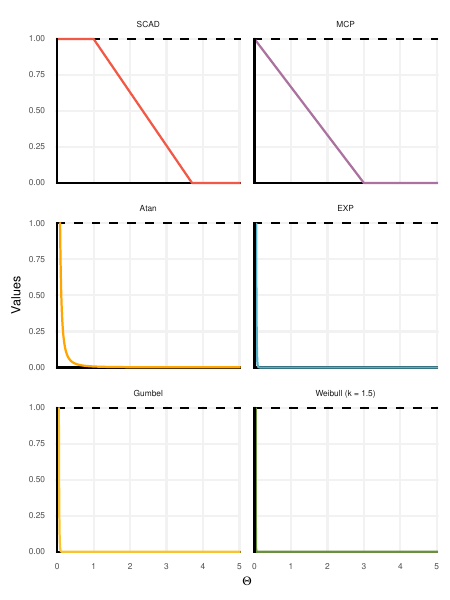}
\caption{\label{fig:derivatives} Depictions of the penalty derivatives in this study organized by non-convex (top row), \(ell_0\) approximators (middle row), and adaptive penalties (bottom row). Hyperparameters were set to \(\lambda = 1\) and defaults for \(\gamma\). Weibull's shape parameter (\(k\)) was set to 1.5 to demonstrate how it influences the derivative relative to EXP. The dashed black line across the top of each panel represnts the \(\ell_1\) penalty and the solid black line down the y-axis and across the x-axis of each panel represents the \(\ell_0\) penalty.}
\end{figure}

\newpage

\begin{flushleft}
\textbf{Table D1} 

\textit{Non-convex Penalties and Derivatives}
\end{flushleft}

\begin{table}[H]
\centering
\renewcommand{\arraystretch}{2.0}
\begin{threeparttable}
\begin{tabular}{ccc}

\textbf{Penalty} & \textbf{\( p_{\lambda,\gamma}(\theta) \)} & \textbf{\( p^\prime_{\lambda,\gamma}(\theta) \)} \\ \hline

Atan & \( \lambda\!\left(\gamma + \dfrac{2}{\pi}\right)\arctan\!\left(\dfrac{|\theta|}{\gamma}\right) \) & \( \dfrac{\lambda\gamma\!\left(\gamma + \frac{2}{\pi}\right)}{\gamma^2 + \theta^2} \) \\ [2em]

MCP & \( \begin{cases}
  \lambda|\theta| - \dfrac{\theta^2}{2\gamma} & |\theta| \leq \gamma\lambda \\[0.4em]
  \dfrac{1}{2}\gamma\lambda^2 & |\theta| > \gamma\lambda
\end{cases} \) & \( \begin{cases}
  \lambda - \dfrac{|\theta|}{\gamma} & |\theta| \leq \gamma\lambda \\[0.4em]
  0 & |\theta| > \gamma\lambda
\end{cases} \) \\ [2em]

SCAD & \( \begin{cases}
  \lambda|\theta| & |\theta| \leq \lambda \\[0.4em]
  \dfrac{2\gamma\lambda|\theta| - \theta^2 - \lambda^2}{2(\gamma - 1)} & \lambda < |\theta| \leq \gamma\lambda \\[0.4em]
  \dfrac{\lambda^2(\gamma + 1)}{2} & |\theta| > \gamma\lambda
\end{cases} \) & \( \begin{cases}
  \lambda & |\theta| \leq \lambda \\[0.4em]
  \dfrac{\gamma\lambda - |\theta|}{\gamma - 1} & \lambda < |\theta| \leq \gamma\lambda \\[0.4em]
  0 & |\theta| > \gamma\lambda
\end{cases} \) \\

\bottomrule
\end{tabular}
\end{threeparttable}
\end{table}

\newpage

\section{Transparency and Openness}\label{appendix-code}

The manuscript was prepared using Rmarkdown (version 2.31; \citealp{rmarkdown1,rmarkdown2,rmarkdown3}) and set in APA style using the \{papaja\} package (version 0.1.4; \citealp{papaja}). Network and data generation were implemented using the \{L0ggm\} package (version 0.1.1; \citealp{L0ggm}). The networks were estimated using the packages that are primarily used in the literature to ensure consistency with applied applications: EBICglasso (\{qgraph\} version 1.9.8; \citealp{qgraph}), BGGM (\{easybgm\} version 0.4.0; \citealp{easybgm}), and non-convex penalties of Atan, MCP, and SCAD (\{GGMncv\} version 2.1.2; \citealp{GGMncv}). The static EXP and data-adaptive penalties were applied using the \{L0ggm\} package. All visualizations were created using the \{ggplot2\} (version 4.0.3; \citealp{ggplot2}), \{ggpubr\} (version 0.6.3; \citealp{ggpubr}), and \{ggrepel\} (version 0.9.8; \citealp{ggrepel}) packages in R. All data generation and simulation scripts are available on the \href{https://osf.io/6hxkt}{\color{blue} OSF}.

\newpage

\section{Empirical Example}\label{appendix-empirical}

\hspace{\parindent}

To investigate the consequences of the differences revealed in the simulation study, Johnson's (\citeyear{johnson2014measuring}) large (\(N\) = 307,313) 300-item IPIP-NEO \citep{goldberg2006international} personality dataset was analyzed for a couple reasons. First, the 300 items of the IPIP-NEO organize into 30 theoretical facets and 5 theoretical factors (Big Five). These 300 items were aggregated into their respective facets using mean scores to create 30 continuous variables, aligning the empirical data with our simulation. Further, because the 30 facets organize into five factors, the theoretical structure of the data parallels our SBM generating mechanism (5 communities and 6 variables per community). Second, because of the large sample size, the empirical asymptotic properties of the estimators can be evaluated using the full sample as a self-referential ground truth. To minimize the influence of potential cultural differences in between-person personality, only participants from the U.S. (\(N\) = 212,265) were used in the analysis.

Two analyses were performed with the empirical data. First, a quantitative and qualitative comparison of EBICglasso, Weibull, and BGGM methods at a single, typical subsample size (\(N\) = 500; \cite{huth2025statistical}) to examine the applied consequences of these different methods. Second, model selection consistency by examining how well each method recovers its own full sample network across subsamples of increasing size. This latter analysis is provided as Supplementary Material (\href{https://osf.io/6hxkt/files/37srg}{\color{blue} OSF}; \cite{christensen2026adaptive_osf}).

\setcounter{figure}{0}
\renewcommand{\thefigure}{F\arabic{figure}}

\begin{landscape}
\begin{figure}[H]
\centering
\includegraphics[width=9in]{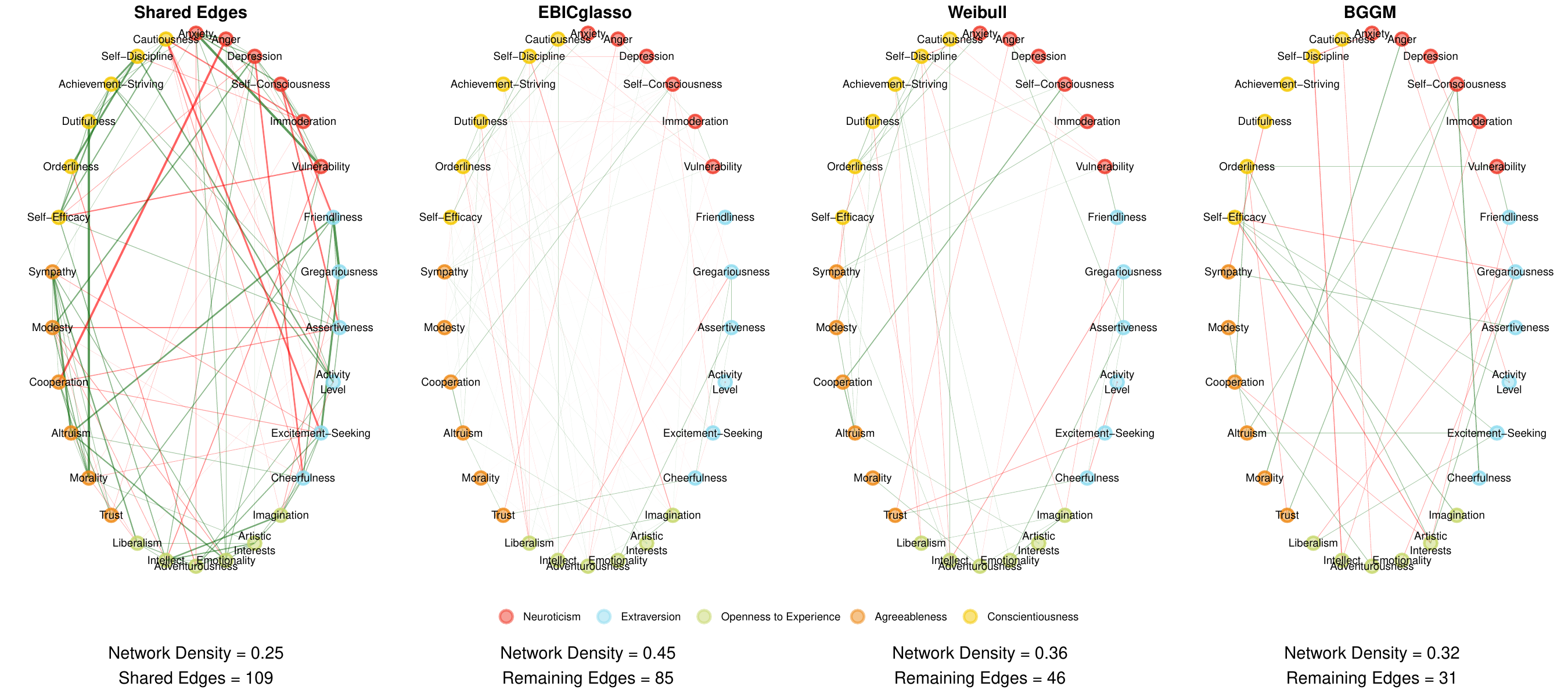}
\caption{\label{fig:empirical_comparison} Network visualizations from a single subsample of \(N\) = 500 respondents from the IPIP-NEO dataset. The leftmost panel displays edges shared across all three methods (EBICglasso, Weibull, and BGGM), with network density and shared edge count reported below. The remaining panels display each method's edges after removing those shared by all three methods. Network density and remaining edge count are reported below each panel. Node color indicates the theoretical Big Five factor membership (Neuroticism, Extraversion, Openness to Experience, Agreeableness, Conscientiousness). Edge color indicates the sign of the partial correlation (green = positive, red = negative) and edge thickness indicates the strength of the partial correlation.}
\end{figure}
\end{landscape}

A single, random subsample of 500 people was drawn from the U.S. respondents in the IPIP-NEO dataset and the EBICglasso, Weibull, and BGGM network estimation methods were applied (Figure \ref{fig:empirical_comparison}). Of the total edges estimated across all three methods, 109 were shared across all three networks. The remaining edges---those unique to a single method or shared between only two---differed systematically by method. EBICglasso produced the densest network (0.45) with the most remaining edges (85), followed by Weibull (0.36 and 46, respectively) and BGGM (0.32 and 31, respectively).

Weibull had no unique edges (i.e., edges not shared with either EBICglasso or BGGM), whereas EBICglasso and BGGM had 42 and 26 unique edges, respectively. Of Weibull's 155 total edges, 151 were shared with EBICglasso and 113 were shared with BGGM, exceeding the 109 edges shared between EBICglasso and BGGM. Together, these results indicate that Weibull occupies a consensus position between the high sensitivity of EBICglasso and high specificity of BGGM (Figure \ref{fig:sen_spec_simulation}).

The centrality rankings corroborated this consensus structure (see Appendix \ref{appendix-empirical}). Overall rank-order agreement (Kendall's \(\tau_b\)), was highest between EBICglasso and Weibull (\(\tau_b\) = 0.720) and substantially lower between Weibull and BGGM (\(\tau_b\) = 0.467) and EBICglasso and BGGM (\(\tau_b\) = 0.426). There was some general consensus among the top five ranked nodes across methods. Anxiety was second highest for EBICglasso and Weibull and seventh for BGGM. Friendliness was fourth for EBICglasso and Weibull and second for BGGM. Cautiousness was second for Weibull, fourth for BGGM, and sixth for EBICglasso.

Several facets, however, exhibited large rank discrepancies across methods and are therefore most consequential for applied interpretation. Assertiveness was ranked first by EBICglasso but fell to sixth and thirteenth under Weibull and BGGM, respectively. The converse pattern held for Excitement-Seeking, which ranked first under BGGM but eleventh under EBICglasso, with Weibull again occupying an intermediate position (fifth). Intellect ranked first for Weibull, third for EBICglasso, and eleventh for BGGM. Artistic Interests showed the starkest divergence: ranked fifth by BGGM but twenty-fifth by EBICglasso and twenty-eighth by Weibull.

These discrepancies carry direct consequences for researchers seeking to identify the single most central node, a practice that remains pervasive in the network psychometrics literature \citep{chambon2026network}. Our results demonstrate that these methods yield different conclusions depending entirely on which method was used: Assertiveness (EBICglasso), Intellect (Weibull), or Excitement-Seeking (BGGM). Notably, in the full sample (\(N\) = 212,265), all three networks converge on Anxiety as the most central node, which provides an empirical basis for evaluating the relative accuracy of each method's subsample rankings.

At \(N\) = 500, EBICglasso and Weibull both ranked Anxiety second whereas BGGM ranked it seventh, suggesting that EBICglasso and Weibull more closely approximated the most central node at this typical subsample size. This advantage is more favorable for Weibull than EBICglasso, however, because Weibull achieves this approximation with a more parsimonious model (i.e., a sparser network).

\newpage

\begin{flushleft}
\textbf{Table F1} 

\textit{Empirical Centrality Rankings from a Single Subsample}
\end{flushleft}

\begin{table}[H]
\centering
\renewcommand{\arraystretch}{0.85}
\begin{threeparttable}
\begin{tabular}{ccc}
\textbf{EBICglasso} & \textbf{Weibull} & \textbf{BGGM}\\
\midrule
\textbf{Assertiveness} & \textbf{Intellect} & \textbf{Excitement-Seeking}\\
Anxiety & Anxiety & Friendliness\\
\textbf{Intellect} & Cautiousness & Self-Efficacy\\
Friendliness & Friendliness & Cautiousness\\
Altruism & \textbf{Excitement-Seeking} & \textbf{Artistic Interests}\\
\addlinespace
Cautiousness & \textbf{Assertiveness} & Cooperation\\
Self-Efficacy & Altruism & Anxiety\\
Self-Discipline & Vulnerability & Depression\\
Depression & Self-Discipline & Sympathy\\
Cooperation & Depression & Morality\\
\addlinespace
\textbf{Excitement-Seeking} & Self-Efficacy & \textbf{Intellect}\\
Vulnerability & Sympathy & Self-Consciousness\\
Morality & Morality & \textbf{Assertiveness}\\
Emotionality & Emotionality & Anger\\
Sympathy & Cooperation & Altruism\\
\addlinespace
Dutifulness & Modesty & Cheerfulness\\
Achievement-Striving & Activity Level & Activity Level\\
Self-Consciousness & Adventurousness & Self-Discipline\\
Cheerfulness & Achievement-Striving & Vulnerability\\
Imagination & Cheerfulness & Modesty\\
\addlinespace
Adventurousness & Anger & Gregariousness\\
Gregariousness & Imagination & Emotionality\\
Anger & Self-Consciousness & Liberalism\\
Activity Level & Dutifulness & Imagination\\
\textbf{Artistic Interests} & Liberalism & Adventurousness\\
\addlinespace
Modesty & Gregariousness & Achievement-Striving\\
Orderliness & Orderliness & Dutifulness\\
Liberalism & \textbf{Artistic Interests} & Orderliness\\
Trust & Trust & Trust\\
Immoderation & Immoderation & Immoderation\\
\hline \addlinespace
\multicolumn{3}{c}{\textbf{Kendall's} \(\boldsymbol{\tau_b}\)} \\
EBICglasso & 0.720 & 0.426 \\
Weibull & --- & 0.467 \\
\bottomrule
\end{tabular}
\begin{tablenotes}
\small
\vspace{0.1em}
\item \textit{Note.} Facets are ranked from highest to lowest node strength estimated from a single random subsample of $N = 500$ US respondents drawn from the IPIP-NEO dataset. Each column presents the centrality ranking for each method. Kendall's $\tau_b$ values in the lower panel reflect pairwise rank-order agreement between each methods' rankings. Bolded facets appear in the top five centrality rankings of at least one method and exhibit the largest rank discrepancies across methods.
\end{tablenotes}
\end{threeparttable}
\end{table}

\newpage

\section{Evaluation Metrics}\label{appendix-metrics}

\subsection{Edge}\label{edge}

Edge recovery was evaluated on whether each network estimation method could successfully recover the correct presence (sensitivity) and absence (specificity) of edges in the simulated population network structure. Sensitivity and specificity are defined as \(\frac{TP}{TP + FN}\) and \(\frac{TN}{TN + FP}\) (respectively) where \(TP\) is an edge present in both the population and estimated networks, \(FN\) is an edge present in the population network but not the estimated network, \(TN\) is an edge absent from both the population and estimated networks, and \(FP\) is an edge present in the estimated network but not the population network.

Edge bias was defined as the mean bias error of TPs (\(\mathrm{MBE}_{TP}\)) and computed to investigate the bias in the asymptotic normality component of the oracle properties:

\[
\mathrm{MBE}_{TP} = \frac{1}{|TP|} \sum_{(i,j) \in TP} (\hat{\omega}_{ij} - \omega_{ij})
\]

\noindent where \(\hat{\omega}_{ij}\) and \(\omega_{ij}\) are the estimated and population partial correlations for edge \((i,j)\), \(TP\) is the set of true positive edges, and \(|TP|\) is their count. Negative values indicate underestimation (shrinkage toward zero).

The median absolute error for FPs and FNs were computed to provide a full picture of edges present in the simulated network: correctly included (TPs), incorrectly included (FPs), and incorrectly excluded (FNs). Figures for FPs and FNs that parallel the the \(MBE_{TP}\) are provided as Supplemental Materials (\href{https://osf.io/6hxkt/files/osfstorage}{\color{blue} OSF}; \citealp{christensen2026adaptive_osf}).

\subsection{Centrality}\label{centrality}

For centrality, node strength (\(NS_i\))---the sum of the absolute values of a node's edge weights---was computed given its prevalence in the literature:

\[
NS_i = \sum_{j=1}^p |\omega_{ij}|.
\]

\noindent Node strength was computed for the population network to obtain population centrality values.

Centrality bias was computed in the same way as \(MBE_{TP}\),

\[
MBE_{NS} = \frac{1}{p} \sum_{i=1}^p (\hat{NS}_i - NS_i),
\]

\noindent where \(\hat{NS}_i\) and \(NS_i\) are the estimated and population node strength for node \(i\), respectively. Centrality bias is not guaranteed to follow the same patterns as \(MBE_{TP}\) because FP edges contribute upward bias (larger estimated node strength) and FN edges contribute downward bias (lower estimated node strength).

Kendall's \(\tau_b\) was computed between estimated and population node strength values to evaluate rank-order congruence. Christensen and Choi (\citeyear{christensen2026network}) noted that, for samples drawn from a bivariate normal population, Kendall's \(\tau_b\) can be expressed as a Pearson's correlation using Greiner's relation \citep{greiner1909fehlersystem}: \(r = \sin\left(\frac{\pi}{2} \tau_b\right).\)

This conversion provides an intuitive interpretation of \(\tau_b\) by anchoring it to familiar Pearson correlation benchmarks. Christensen and Choi (\citeyear{christensen2026network}) recommended \(r = 0.90\) as a threshold for adequate recovery, corresponding to \(\tau_b = 0.713\) and a concordance of approximately 85\% of pairs. We adapted this framework using a graded classification: \(r = 0.70\) (adequate), \(r = 0.80\) (good), and \(r = 0.90\) (excellent), corresponding to \(\tau_b = 0.494\), \(0.590\), and \(0.713\), and concordance of approximately 75\%, 80\%, and 85\% of pairs, respectively. To simplify interpretability, we used values 0.50, 0.60, and 0.70 as adequate, good, and excellent rank-order congruence.

\end{document}